\documentclass[envcountsame,envcountsect]{llncs}
\usepackage{wrapfig}
\usepackage{graphicx}
\usepackage{graphicx}
\usepackage{amssymb} 
\usepackage{amsmath} 
\usepackage{algorithmic}

\newcommand{\nmodels}{\ensuremath{\not\models}}
\newcommand{\mathtext}[1]{\ensuremath{\mathrm{\text{#1}}}}
\newcommand{\set}[1]{\ensuremath\left\{#1\right\}}

\newcommand{\complexityclassname}[1]{\ensuremath{\mathrm{#1}}}
\newcommand{\complementclass}[1]{\complexityclassname{co}#1}
\newcommand{\PTIME}{\complexityclassname{P}}
\newcommand{\PSPACE}{\complexityclassname{PSPACE}}
\newcommand{\NP}{\complexityclassname{NP}}
\newcommand{\LOGSPACE}{\complexityclassname{LOGSPACE}}
\newcommand{\CONP}{\complementclass{\NP}}

\newcommand{\redlogm}{\ensuremath{\leq_{m}^{\log}}}
\newcommand{\eqlogm}{\ensuremath{\equiv_{m}^{\log}}}
\newcommand{\redpm}{\ensuremath{\leq_{m}^{p}}}
\newcommand{\redeqpm}{\ensuremath{\equiv_{m}^{p}}}
\newcommand{\prblemname}[1]{\ensuremath{\mathsf{#1}}}

\newcommand{\clonename}[1]{\mathrm{#1}}
\newcommand{\cM}{\ensuremath{\clonename{M}}}
\newcommand{\cR}{\ensuremath{\clonename{R}}}
\newcommand{\cBF}{\ensuremath{\clonename{BF}}}
\newcommand{\cS}{\ensuremath{\clonename{S}}}
\newcommand{\cD}{\ensuremath{\clonename{D}}}
\newcommand{\cV}{\ensuremath{\clonename{V}}}
\newcommand{\cE}{\ensuremath{\clonename{E}}}
\newcommand{\cL}{\ensuremath{\clonename{L}}}
\newcommand{\cI}{\ensuremath{\clonename{I}}}
\newcommand{\cN}{\ensuremath{\clonename{N}}}

\newcommand{\clone}[1]{\ensuremath{\left[ #1\right]}}
\newcommand{\dual}[1]{\mathrm{dual}\left(#1\right)}
\newcommand{\card}[1]{\left| #1 \right|}
\newcommand{\pred}{\ensuremath{\text{pred}}} 
\newcommand{\wpred}{\ensuremath{\wedge\text{-pred}}}
\newcommand{\md}[1]{\text{\it md}\!\left(#1\right)}
\newcommand{\smd}[1]{\text{\scriptsize{\it md}}\left(#1\right)}

\newcommand{\algspace}{\vspace*{1mm}}
\newcommand{\multform}[3]{\mathrm{MFORM}_{#1}^{#2}\left(#3\right)}
\newcommand{\multcirc}[3]{\mathrm{MCIRC}_{#1}^{#2}\left(#3\right)}
\newcommand{\formsat}[4]{\ensuremath{#1\mathtext{-}\prblemname{FSAT}_{#2}^{#3}\left(#4\right)}}
\newcommand{\circsat}[4]{\ensuremath{#1\mathtext{-}\prblemname{CSAT}_{#2}^{#3}\left(#4\right)}}
\newcommand{\formtaut}[4]{\ensuremath{#1\mathtext{-}\prblemname{FTAUT}_{#2}^{#3}\left(#4\right)}}
\newcommand{\circtaut}[4]{\ensuremath{#1\mathtext{-}\prblemname{CTAUT}_{#2}^{#3}\left(#4\right)}}

\newcommand{\logicname}[1]{\ensuremath{\text{\upshape{\ensuremath{\mathrm{#1}}}}}}
\newcommand{\K}{\logicname{K}}
\newcommand{\T}{\logicname{T}}
\newcommand{\KD}{\logicname{KD}}
\newcommand{\Sfour}{\logicname{S4}}
\newcommand{\Sfive}{\logicname{S5}}
\newcommand{\Kfour}{\logicname{K4}}

\newcommand{\algname}{\ensuremath{\oplus\text{\sc -Sat}}}

\newcommand{\algorithmicprocedure}[1]{\textbf{Procedure #1}}
\newcommand{\algorithmicendproc}{\algorithmicend\ \algorithmicprocedure{}}

\newenvironment{ALC@proc}{\begin{ALC@g}}{\end{ALC@g}}
\newcommand{\BEGINPROC}[1][default]{\ALC@it\algorithmicprocedure%
\ALC@com{##1}\begin{ALC@proc}}
  \ifthenelse{\boolean{ALC@noend}}{
    \newcommand{\ENDIF}{\end{ALC@if}}
    \newcommand{\ENDFOR}{\end{ALC@for}}
    \newcommand{\ENDWHILE}{\end{ALC@whl}}
    \newcommand{\ENDLOOP}{\end{ALC@loop}}
    \newcommand{\ENDPROC}{\ed{ALC@proc}}
  }{
    \newcommand{\ENDPROC}{\end{ALC@proc}\ALC@it\algorithmicendproc}
  }

\title{Generalized Modal Satisfiability
\thanks{Supported in part by the DAAD Postdoc Program, by grants NSF-CCR-0311021, NSF-IIS-0713061, and DFG VO 630/5-1, and by a Friedrich Wilhelm Bessel
Research Award. Work done in part while the second and third authors worked at the Leibniz Universit\"at Hannover. An earlier version of some of the results appeared as~\cite{bhss05b}.}}

\begin{document}

\pagestyle{plain}

\author{%
  Edith Hemaspaandra
  \and
  Henning Schnoor
  \and
  Ilka Schnoor
}

\institute{%
  Department of Computer Science,
  Rochester Institute of Technology,
  Rochester, NY 14623, U.S.A.
  \email{\{eh,hs,is\}@cs.rit.edu}
} 

\bibliographystyle{alpha}
\maketitle

\begin{abstract}
It is well known that modal satisfiability is \PSPACE-complete \cite{lad77}. However, the complexity may decrease if we restrict the set of propositional operators used. Note that there exist an infinite number of propositional operators, since a propositional operator is simply a Boolean function. We completely classify the complexity of modal satisfiability for every finite set of propositional operators, i.e., in contrast to previous work, we classify an infinite number of problems. We show that, depending on the set of propositional operators, modal satisfiability is \PSPACE-complete, \CONP-complete, or in \PTIME.  We obtain this trichotomy not only for modal formulas, but also for their more succinct representation using modal circuits. We consider both the uni-modal and the multi-modal case, and study the dual problem of validity as well.

\bigskip

\end{abstract}

\section{Introduction}

Modal logics are valuable tools in computer science, since they are often a good compromise between expressiveness and decidability. Standard applications of modal logics are in artificial intelligence \cite{moo79,mcsahaig78}, and cryptographic and other protocols \cite{frhuje02,codofl03,hamotu88,lare86}. More recent applications include a new modal language called Versatile Event Logic \cite{bega04}, and the usage to characterize the relationship among belief, information acquisition, and trust \cite{lia03}. 

Applications of modal logic for solving practical problems obviously require a study of the computational complexity of various aspects of modal logics. A central computational problem related with any logic is the satisfiability problem, that is to decide whether a given formula has a model. The first complexity results for the modal satisfiability problem were achieved by Ladner~\cite{lad77}. He showed that the basic modal satisfiability problem is \PSPACE-complete. There is a rich literature on the complexity of variants of the modal satisfiability problem, important works include the paper by Halpern and Moses \cite{hamo92} on multi-modal logics. Recently, \PSPACE-algorithms for a wide class of modal logics were presented by Schr\"oder and Pattinson~\cite{schpa06}.

For modal logics to be used in practice, a lower complexity of the satisfiability problem than the aforementioned \PSPACE-hardness is desirable. It turns out that for many applications, the full power of modal logic is not necessary. There are various ways of defining restrictions of modal logics which potentially lead to a computationally easier version of the satisfiability problem that have been studied: Variations of modal logics are achieved by restricting the class of considered models, e.g., instead of allowing arbitrary graphs, classical examples of logics only allow reflexive, transitive, or symmetric graphs as models. Many complexity results for logics defined in this way have been achieved: Initial results for many important classes are present in the above-mentioned work by Ladner~\cite{lad77}. Recently, Hemaspaandra and Schnoor considered a uniform generalization of many of these examples~\cite{hem-schnoor-08:elementary-modal-logics-STACS}. It should be noted that such restrictions do not necessarily decrease the complexity; for many common restrictions, the complexity remains the same~\cite{lad77,hamo92} and it is even possible that the complexity increases. In \cite{hem96}, Hemaspaandra showed that the complexity of the global satisfiability problem increases from EXPTIME-complete to undecidable by restricting the graphs to those in which every node has at least two successors and at most three 2-step successors.

Another way of restricting modal logics is to change the syntax rather than the semantics, i.e., restrict the structure of the considered modal formulas. Syntactical restrictions are known to naturally reduce the complexity of many decision problems in logic. In propositional logic, well-known examples are the satisfiability problems for Horn formulas, 2CNF formulas, or formulas describing monotone functions: All of these can be solved in polynomial time, while the general propositional satisfiability problem is \NP-complete. Syntactical restrictions have been considered in the context of modal logics before: Halpern showed that the complexity of the modal satisfiability problem decreases to linear time when restricting the number of variables and nesting degree of modal operators \cite{hal95}. Restricted modal languages where only a subset of the relevant modal operators are allowed have been studied in the context of linear temporal logic (see, e.g.,~\cite{sicl85}). Some description logics can be viewed as modal logic with
a restriction on the propositional operators that are allowed.
For the complexity of description logics, see, e.g.,~\cite{ss91,dhlnns92,dlnn97}.
For the complexity of modal logic with other restrictions on the set
of operators, see~\cite{hem01}.

The approach we take in the present paper is to generalize the occurring propositional operators in the formulas. Instead of the operators $\wedge,\vee$ and negation, we allow the appearing operators to represent arbitrary Boolean functions. In particular, there are an infinite number of Boolean operators. We completely classify the complexity of modal satisfiability for every
finite set of propositional operators. The restriction on the propositional 
operators leads to a classification following the structure of Post's Lattice
\cite{pos41}, a tool that has been applied in similar contexts before: For propositional logic, Lewis showed that the satisfiability problem is dichotomic: Depending on the set of operators, propositional satisfiability
is either \NP-complete or solvable in polynomial time \cite{lew79}. For modal 
satisfiability, we achieve a trichotomy: For the modal logic \K, the satisfiability problem is \PSPACE-complete, \CONP-complete, or in \PTIME. We also achieve a full classification for the logic \KD\ (in this case, we show a \PSPACE/\PTIME-dichotomy), and almost complete classifications for the logics \T, \Sfour, and \Sfive.

When considering sets of operations which do not include negation, the
complexity for the cases where one modal operator is allowed
sometimes differs from the case where we allow both operator
$\Diamond$ and its dual operator $\Box$. With only one of these,
modal satisfiability is \PSPACE-complete exactly in those cases
in which propositional satisfiability is \NP-complete. When we
allow both modal operators, the jump to \PSPACE-completeness happens
earlier, i.e., with a set of operations with less expressive power.

We consider several generalizations of the problems outlined above. In particular, we introduce \emph{modal circuits} as a succinct way of representing modal formulas. We show that this does not give us a significantly different complexity than the formula case. We also consider multi-modal logics, in which several independant modal operators are introduced.

In addition to the satisfiability problem, we also study the validity problem, where we do not ask whether a formula is satisfiable, but whether it is true in every possible model. Since our restricted modal languages do not always include negation, the complexity of this problem turns out to be different from, but related to, the complexity of the satisfiability problem.

An interesting case in our classifications is the case where we only allow the propositional exclusive-or and constants as propositional operators. For purely propositional logics, it is very easy to see that satisfiability for these formulas (essentially linear equations over GF(2)) can be decided in polynomial time. In the case of modal logics, an analogous result holds, but the proof requires significantly more work. As in the propositional case, it yields an optimal solution to the \emph{minimization problem} as well: Given a modal formula or modal circuit using only these propositional operators, we can efficiently compute an equivalent formula or circuit of minimal size.

The structure of the paper is as follows: In Section~\ref{section:preliminaries}, we introduce the necessary definitions, recall results from the literature, and prove some basic facts about our problems. Section~\ref{section:satisfiability} contains our main results: The complete classification of the complexity of the modal satisfiability problem for every possible set of Boolean operators. In Section~\ref{section:validity} we prove a relationship between satisfiability and validity implying a full classification of this problem as well. We conclude in Section~\ref{section:conclusion} with some open questions for future research.

\section{Preliminaries}\label{section:preliminaries}

\subsection{Modal Logic}

Modal logic is an extension of classical propositional logic that talks about ``possible worlds.'' We first introduce the usual uni-modal logic, and then generalize it to the multi-modal case. Uni-modal logics enrich the vocabulary of propositional logic with an additional unary modal operator $\Diamond.$ A model for a given formula consists of a directed graph with propositional assignments. To be more precise, a \emph{frame} consists of a set $W$ of ``worlds,'' and a ``successor'' relation $R\subseteq W\times W$. For $(w,w')\in R$, we say $w'$ is a \emph{successor} of $w$. A \emph{model} $M$ consists of a frame $(W,R)$, a set $X$ of propositional variables, and a function $\pi\colon X\rightarrow {\cal P}(W).$ The intuition is that for $x\in X,$ $\pi(x)$ denotes the set of worlds in which the variable $x$ is true. The operator $\Box$ is the dual operator to $\Diamond,$ $\Box\varphi$ is defined as $\neg\Diamond\neg\varphi$. Intuitively, $\Diamond\varphi$ means ``there is a successor world in which $\varphi$ holds,'' and $\Box\varphi$ means ``$\varphi$ holds in all successor worlds.'' For a class $\cal F$ of frames, we say a model $M$ is an $\cal F$-model if the underlying frame is an element of $\cal F$.

In multi-modal logic, a finite number of these modal operators is considered, where each operator $\Diamond_i$ corresponds to an individual successor relation $R_i.$ For a modal logic with $k$ modalities, a frame again consists of a set $W$ of worlds, and successor relations $R_1,\dots,R_k\subseteq W\times W.$ If $(w,w')\in R_i,$ we say that $w'$ is a \emph{$i$-successor} of $w.$ For a formula $\varphi$ built over the variables $X,$ propositional operators $\wedge$ and $\neg,$ and the unary modal operators $\Diamond_1,\dots,\Diamond_k,$ we define what ``$\varphi$ holds at world $w$'' means for a model $M$ (or $M,w$ \emph{satisfies} $\varphi$) with assignment function $\pi,$ written as $M,w\models\varphi$.

\begin{itemize}
\item If $\varphi$ is a propositional variable $x,$ then $M,w\models\varphi$ if and only if $w\in\pi(x)$,
\item $M,w\models\varphi_1\wedge\varphi_2$ if and only if ($M,w\models\varphi_1$ and $M,w\models\varphi_2$),
\item $M,w\models\neg\varphi$ if and only if $M, w \not \models \varphi$, 
\item for $i\in\set{1,\dots,k}$, $M,w\models\Diamond_i\varphi$ if and only if there is a world $w'\in W$ such that $(w,w')\in R_i$ and $M,w'\models\varphi$.
\end{itemize}

Analogously to the unimodal case, the operator $\Box_i$ is defined as $\Box_i\varphi=\neg\Diamond_i\neg\varphi$. For a class $\mathcal F$ of frames, we say a formula $\varphi$ is
$\mathcal F$-satisfiable if there exists an $\cal F$-model $M=(W,R, \pi)$ and
a world $w\in W$ such that $M,w\models\varphi$. For modal formulas $\varphi$ and $\psi,$
we write $\varphi\equiv_{\cal F} \psi$ if for every world in every
$\cal F$-model, $\varphi$ holds if and only if $\psi$ holds.
Note that a formula $\varphi$ is $\cal F$-satisfiable iff $\varphi \not\equiv_{\cal F} 0$. Similarly, we say that $\varphi$ is an $\cal F$-tautology if $\varphi\equiv_{\cal F} 1$, and 
finally $\varphi$ is $\cal F$-constant if $\varphi\equiv_{\cal F} 0$ or $\varphi\equiv_{\cal F} 1$.

\begin{center}
\begin{table}
  \begin{tabular}{ll|ll}
    \K & \hspace*{1mm} & \hspace*{1mm} & All frames \\ \hline
    \KD &&& Frames in which every world has a successor \\ \hline
    \Kfour &&& Transitive frames \\ \hline
    \Sfour &&& Frames that are reflexive and transitive \\ \hline
    \Sfive &&& Frames that are reflexive, transitive, and symmetric \\ \hline
    \T &&& Reflexive frames \\
  \end{tabular}
   \caption{Classes of frames}
\end{table}
\end{center}
We now define the classes of frames that are most commonly used in applications
of modal logic.  To see how these frames correspond to
axioms and proof systems, see, for example,~\cite[Section 4.3]{blrive01}. Again, we first consider the uni-modal case and then present the natural generalizations to multi-modal logics. \K\ is the class of all frames, \KD\ is the class of frames in which every
world has a successor, i.e., for all $w\in W$, there is a $w'\in W$
such that $(w,w')\in R$. \T\ is the class of reflexive frames, 
\Kfour\ is the class of transitive frames,
\Sfour\ is the class of frames that are both reflexive and transitive, and \Sfive\ is the class of
reflexive, symmetric, and
transitive frames. The \emph{reflexive singleton} is the frame consisting
of one world $w$, and the relation $\{(w,w)\}.$ Note that all classes of frames $\cal F$ described above contain the reflexive singleton. Similarly, the \emph{irreflexive singleton} is the frame consisting of one world, and an empty successor relation.

For multi-modal logics, the generalizations are obvious: For a class of frames $\mathcal F$ as previously defined, we say that the class $\mathcal F_k$ contains those frames $(W,R_1,\dots,R_k),$ where $(W,R_i)\in\mathcal F$ for all $i\in\set{1,\dots,k}.$  In particular, a multi-modal reflexive singleton consists of the set of worlds $W=\set{w}$ where each successor relation consists of the pair $(w,w),$ and the multi-modal irreflexive singleton consists of the same set of worlds where all of the successor relations are empty. If the number $k$ of modal operators is clear from the context, we often simply write $\mathcal F$ instead of $\mathcal F_k,$ speak about the reflexive singleton, etc.

\subsection{Generalized Formulas and Circuits}

We now consider a more general notion of modal formulas, whose propositional analog has been studied extensively. We generalize the notion of a modal formula in two ways: First, instead of allowing the usual propositional operators $\wedge,\vee$, and $\neg$, we allow arbitrary Boolean functions. Second, we study circuits as succinct representations of formulas. Intuitively, a circuit is a generalization of a formula in the same way as a directed acyclic graph is a generalization of a tree, since formulas directly correspond to tree-like circuits. To be more precise, for a finite set $B$ of Boolean functions, a \emph{modal $B$-circuit} is a generalization of a propositional Boolean circuit (see e.g., \cite{vol99} for an introduction to Boolean circuits) with gates for functions from $B$ and additional gates representing the modal operators $\Diamond_i$ or $\Box_i.$ Boolean circuits are a standard way to succinctly represent Boolean functions. Formally, we define the following (recall that $X$ is the set of propositional variables):

\begin{definition}
Let $B$ be a finite set of Boolean functions, and let $M\subseteq\set{\Box,\Diamond}.$ A circuit in $\multcirc MkB$ is a tuple $C=(V,E,\alpha,\beta,\mathit{out})$ where $(V,E)$ is a finite directed acyclic graph, $\alpha\colon E\rightarrow\mathbb{N}$ is an injective function, $\beta\colon V\rightarrow B\cup\{\Box_1,\dots,\Box_k,\Diamond_1,\dots,\Diamond_k\}\cup X$ is a function, and $\mathit{out}\in V$, such that
\begin{itemize}
\item{If $v\in V$ has in-degree $0$, then $\beta(v)\in X$ or $\beta(v)$ is a $0$-ary function (a constant) from $B$.}
\item{If $v\in V$ has in-degree $1$, then $\beta(v)$ is a $1$-ary function from $B$ or, for some $i\in\set{1,\dots,k},$ one of the operators $\Box_i$ (if $\Box\in M$) or $\Diamond_i$ (if $\Diamond\in M$).}
\item{If $v\in V$ has in-degree $d>1$, then $\beta(v)$ is a $d$-ary function from $B$.}
\end{itemize}
\end{definition}

By definition, $\multcirc MkB$ contains the modal circuits that use the following as operators: functions from $B$, the modal operators $\Diamond_1,\dots,\Diamond_k$ if $\Diamond\in M$, and $\Box_1,\dots,\Box_k$ if $\Box\in M$. Nodes $v\in V$ are called \textit{gates} of $C$, $\beta(v)$ is the \emph{gate-type} of $v$. The node $\mathit{out}$ is the \textit{output-gate} of $C$. The function $\alpha$ is needed to define the order of arguments for non-commutative functions. The \emph{size} of a modal circuit $C$ is the number of gates: $\card{C}:=\card{V}$. 

In addition to circuits, we also study the special case of modal formulas. A \emph{modal $B$-formula} is a modal $B$-circuit where each gate has out-degree $\leq 1$. This corresponds to the intuitive idea of a formula: Such a circuit can be written down as a formula, e.g., in prefix notation, without growing significantly in size. Semantically we interpret a circuit as a succinct representation of its formula expansion. For a modal $B$-circuit $C$, the \emph{modal depth of $C,$} $\md C$, is the maximal number of gates representing modal operators on a directed path in the graph. If there are no modal gates (i.e., gates $v\in C$ such that $\beta(v)\in\{\Box_i,\Diamond_i\}$ for any $i$) then $\varphi_C$ is a \emph{propositional Boolean formula} and $C$ is a \emph{propositional Boolean circuit}. 

In order to define the semantics of the circuits defined above, we relate them to formulas in the following natural way: The circuit $C$ represents the modal formula $\varphi_C$ that is inductively defined by a modal $B$-formula $\varphi_v$ for every gate $v$ in $C:$

\begin{definition}
\begin{itemize}
\item{If $v\in V$ has in-degree $0$, then $\varphi_v:=\beta(v).$}
\item{Let $v\in V$ have in-degree $l>0$, and let $v_1,\dots,v_l$ be the predecessor gates of $v$ such that $\alpha((v_1,v))<\dots<\alpha((v_k,v))$. Then let $\varphi_v:= \beta(v)(\varphi_{v_1},\dots,\varphi_{v_l})$.}
\item{Finally, we define $\varphi_C$ as $\varphi_{\mathit{out}}$. We call $\varphi_C$ the \emph{formula expansion of $C$}.}
\end{itemize}
\end{definition}

Since every Boolean function can be expressed using only conjunction and negation, the semantics for circuits allowing arbitrary Boolean functions is immediate. It is obvious from the definition that for every modal circuit, there is an equivalent formula. Therefore, considering circuits instead of formulas does not increase the expressive power, but circuits are a succinct representation of formulas (there are circuits representing formulas where the size of the formula is exponential in the size of the circuit).

\subsection{Problem Definitions}

We now define the various modal satisfiability problems we are interested in. As usual in computational complexity, we define the problems as the sets of their yes-instances. 

\begin{definition}
  Let $B$ be a finite set of Boolean functions, $\mathcal F$ a class of frames, $k\ge 0$, and $M\subseteq\{\Diamond,\Box\}$. Then 
  
  \begin{itemize}
  \item $\multform MkB$ is the set of formula expansions of circuits in $\multcirc MkB$, i.e., the set of modal formulas using operators from $B$, and modalities $\Box_1,\dots,\Box_k$ (if $\Box\in M$) and $\Diamond_1,\dots,\Diamond_k$ (if $\Diamond\in M$).
  \item $\formsat{\mathcal F}{M}{k}{B}$ is the set of $\mathcal F_k$-satisfiable formulas from $\multform{M}{k}{B}$.
  \item $\circsat{\mathcal F}{M}{k}{B}$ is the set of $\mathcal F_k$-satisfiable circuits from $\multcirc{M}{k}{B}$.
  \item $\formtaut{\mathcal F}{M}{k}{B}$ is the set of $\mathcal F_k$-tautologies in  $\multform{M}{k}{B}$,
  \item $\circtaut{\mathcal F}{M}{k}{B}$ is the set of $\mathcal F_k$-tautologies in  $\multcirc{M}{k}{B}$.
  \end{itemize}
 \end{definition}

For readability, we often leave out the set brackets and write, for example, $\formsat{\K}{\Box}{1}{\oplus,1}$ instead of $\formsat{\K}{\set{\Box}}{1}{\set{\oplus,1}}$. In addition to specifying whether $\Diamond$ and $\Box$ are allowed ``globally,'' we could also allow our model to specify for each $i\in\set{1,\dots,k}$ whether $\Diamond_i$ and $\Box_i$ are allowed to appear in the circuits. However, our hardness results usually require only a single one of these operators to be present (and upper complexity bounds obviously transfer to the restricted setting). Therefore, the definition we gave captures the significant variations of the problems we study.

From the definitions, the following is immediate, which we will often use without reference. It is obvious that analogous results hold for the tautology problem as well. Due to this proposition, it is clear that it suffices to state lower complexity bounds for the problems involving formulas, and upper bounds for the problems involving circuits.

\begin{proposition}\label{prop:trivial reductions}
Let $B_1\subseteq B_2$ be finite sets of Boolean functions, let $\mathcal F$ be a class of frames, let $k_1\leq k_2$, and let $M_1\subseteq M_2\subseteq\set{\Box,\Diamond}.$ Then the following hold:
\begin{itemize}
\item $\formsat{\mathcal F}{M_1}{k_1}{B_1}\redlogm\formsat{\mathcal F}{M_2}{k_2}{B_2},$
\item $\formsat{\mathcal F}{M_1}{k_1}{B_1}\redlogm\circsat{\mathcal F}{M_2}{k_2}{B_2},$
\item $\circsat{\mathcal F}{M_1}{k_1}{B_1}\redlogm\circsat{\mathcal F}{M_2}{k_2}{B_2}.$
\end{itemize}
\end{proposition}

Initial complexity results can be found in the literature; we state them in our notation:

\begin{theorem}[\cite{hamo92},\cite{lad77}]
\begin{enumerate}
\item $\formsat{\Sfive}{\Box}{1}{\wedge,\neg}$ is \NP-complete.
\item Let $\mathcal F\in\set{\K,\KD,\Kfour,\T,\Sfour}.$ Then $\formsat{\mathcal F}{\Box}{1}{\wedge,\neg}$ is \PSPACE-complete.
\item Let $\mathcal F\in\set{\K,\KD,\Kfour,\T,\Sfour,\Sfive},$ and let $k\ge 2.$ Then $\formsat{\mathcal F}{\Box}{k}{\wedge,\neg}$ is \PSPACE-complete.
\end{enumerate}
\end{theorem}

In \cite{hem01}, Hemaspaandra examined the complexity of $\formsat{\K}{M}{1}{B}$ 
for all $M\subseteq\{\Box,\Diamond\}$ and $B\subseteq\{\wedge,\vee,\neg,0,1\}.$ In this paper, we generalize this result in several ways: We classify the complexity of modal satisfiability for all finite sets of Boolean functions (in particular, we determine the complexity of an infinite number of problems), and we consider multi-modal logic as well. Further, we also consider the case of circuits instead of formulas, and study different frame classes. Finally, we also consider the validity problem.

\subsection{Clones and Post's Lattice}

The notion of clones is very helpful to bring structure to this infinite set of problems. We introduce the necessary definitions, and some important properties of Boolean functions. An $n$-ary function $f$ is a \emph{projection function} if there is some $i$ such that for all $\alpha_1,\dots,\alpha_n\in\set{0,1}$, $f(\alpha_1,\dots,\alpha_n)=\alpha_i$. A set $B$ of Boolean functions is called a \emph{clone} if it is closed under \emph{superposition}, that is, $B$ contains all projection functions and is closed under permutation of variables, identification of variables, and arbitrary composition. It is easy to see that the set of clones forms a lattice. Post determined the complete set of clones, as well as their inclusion structure \cite{pos41}. A graphical presentation of the lattice of clones, also known as Post's Lattice, can be found in Figure~\ref{fig:lattice}. For a set $B$ of Boolean functions, let $\clone B$ be the smallest clone containing $B.$ 

We briefly define the clones that arise in our complexity classification. The smallest clone contains only projections and is named $\mathtext{I}_2.$ Further, $\mathtext{I}_1 = \clone{\{1\}}$. The largest clone $\mathtext{BF}=\clone{\{\wedge, \neg\}}$ is the set of all Boolean functions. The set of all monotone functions forms a clone denoted by $\mathtext{M}=\clone{\{\vee,\wedge,0,1\}}.$ \cD\ consists of all {\em self-dual} functions, i.e., $f\in\cD$ if and only if $f(x_1,\dots,x_n)=\neg f(\overline x_1,\dots,\overline x_n).$ $\cL=\clone{\{\oplus,1\}}$ is the set of all linear Boolean functions (where $\oplus$ is the Boolean exclusive or). The clone of all Boolean functions that can be written using only disjunction and constants is called $\mathtext{V}=\clone{\{\vee,1,0\}}$; further, $\mathtext{V}_0=\clone{\{\vee,0\}}$ and $\mathtext{V}_2=\clone{\{\vee\}}$. Similarly, the clone $\mathtext{E}=\clone{\{\wedge,0,1\}}$ contains the Boolean functions that can be written as conjunctions of variables and constants; $\mathtext{E}_0=\clone{\{\wedge, 0\}}$ and $\mathtext{E}_2=\clone{\{\wedge\}}.$ $\mathtext{R}_1$ is built from all \emph{1-reproducing} functions, i.e., all functions $f$ satisfying $f(1,\dots,1)=1.$ The clone $\cN = \clone{\{\neg, 1\}}$ consists of the projections, their negations, and all constant Boolean functions. $\cS_1 = \clone{\{x \wedge \overline{y}\}}$ and $\mathtext{S}_{11}=\mathtext{S}_1\cap\mathtext{M}.$

\begin{figure}
  \begin{tabular}{lc|cl}
    $\cBF$ & \hspace*{2.5mm} & \hspace*{2.5mm} &  All Boolean functions \\
    \hline
    $\cS_1$ &&& $\clone{x\wedge\overline{y}}$ \\
    \hline
    $\cM$ &&& Monotone functions \\
    \hline
    $\cS_{11}$ &&& $\cM\cap\cS_1$ \\
    \hline
    $\cR_1$ &&& $f$ with $f(1,\dots,1)=1$ \\
    \hline
    $\cD$ &&& Self-dual functions \\
    \hline
    $\cL$ &&& Linear functions \\
    \hline
    $\cV$ &&& Multi-ary OR and constants $0$, $1$ \\
    \hline
    $\cV_0$ &&& Multi-ary OR and constant $0$ \\
    \hline
    $\cV_2$ &&& Multi-ary OR \\
    \hline
    $\cE$ &&& Multi-ary AND and constants $0$, $1$ \\
    \hline
    $\cE_0$ &&& Multi-ary AND and constant $0$ \\
    \hline
    $\cE_2$ &&& Multi-ary AND \\
    \hline
    $\cN$ &&& Negation, idendity, and constants \\
    \hline
    $\cI$ &&& Identity and constants
  \end{tabular}

\end{figure}

\begin{figure}
  \hspace*{-0.525cm}\resizebox{150mm}{190mm}{
  \includegraphics{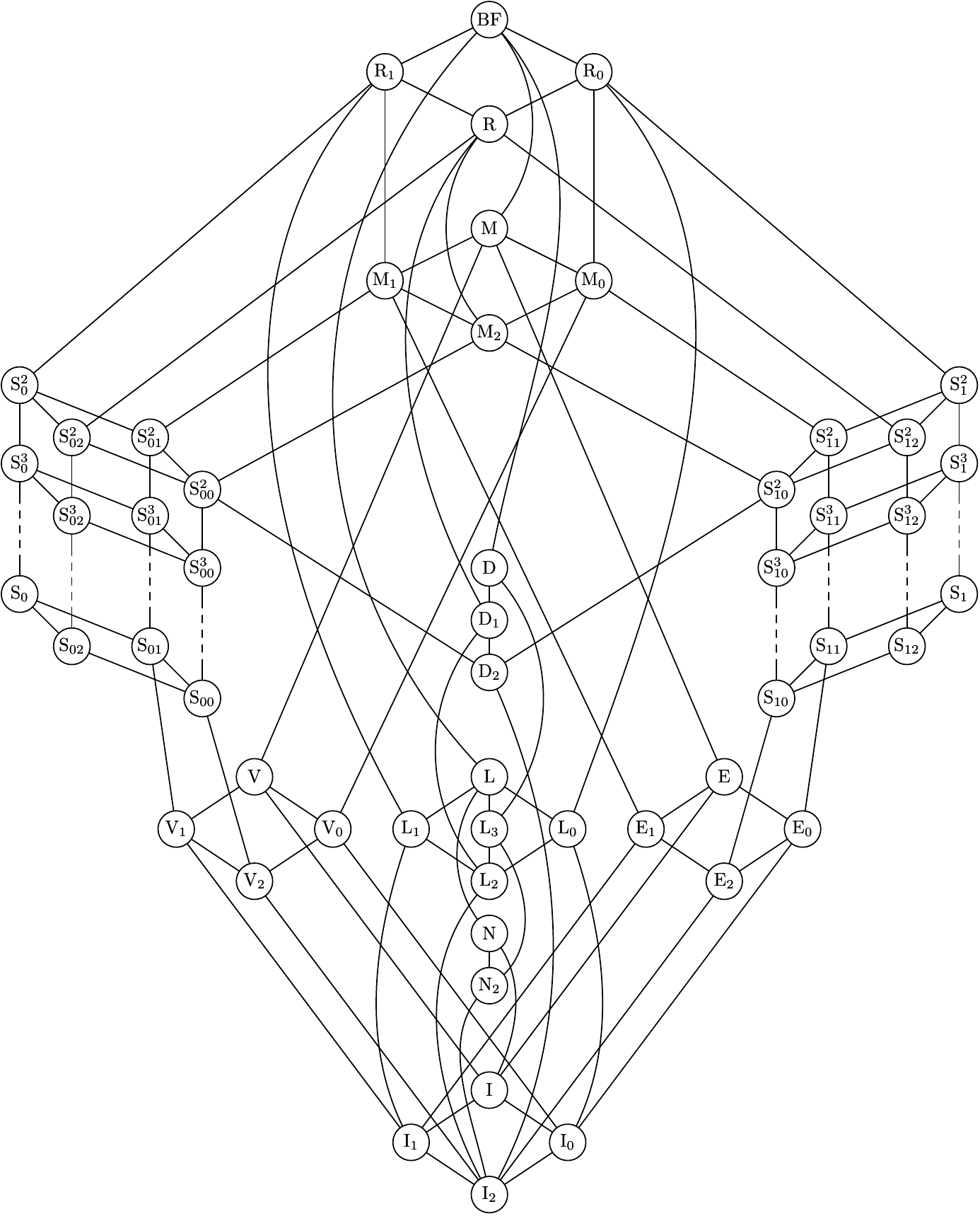}}
  \caption{Post's lattice}
  \label{fig:lattice}
\end{figure}

 If we interpret Boolean formulas as Boolean functions, then
\clone B consists of all propositional formulas that are
equivalent to a formula built with variables and operators from $B$.
Therefore, this framework can be used to
investigate problems related to Boolean formulas depending on which
connectives are allowed. Several problems have been studied in this context: Lewis proved that the
satisfiability problem for Boolean formulas with connectives from $B$
is \NP-complete if $\mathtext{S}_1\subseteq \clone B$ and in \PTIME\
otherwise \cite{lew79}. Another example is the classification of the
equivalence problem given by Reith: Deciding whether
two formulas with connectives from $B$ are equivalent is in \LOGSPACE\ if 
$\clone B \subseteq \mathtext{V}$ or $\clone B \subseteq \mathtext{E}$ 
or $\clone B \subseteq \mathtext{L}$, and \CONP-complete in all other
cases \cite{rei01}. Dichotomy results for counting the solutions 
of formulas \cite{rewa99-appeared}, finding the minimal solutions of formulas
\cite{revo00}, and learnability of Boolean formulas and circuits \cite{dal00} 
were achieved as well. After presenting our results in \cite{bhss05b}, analogous classifications have been achieved by Bauland et al.\ in the context of temporal logics \cite{bsssv07,bmsssv07:toappear}.

Post's Lattice has also been a helpful tool in the constraint satisfaction
context. It can be used to obtain a very easy proof of Schaefer's Theorem \cite{sch78} and related complexity classifications. This is surprising, because constraint satisfaction problems are not related
to Post's Lattice by definition, but clones appear indirectly 
through a Galois connection~\cite{jecogy97}. For more information about the use of Post's Lattice in complexity classifications of propositional logic, see, for example,~\cite{bcrv03,bcrv04}. Finally, the notion of clones is not restricted to the Boolean case, but has been studied for arbitrary domains. The monograph \cite{lau06:clonebook} is an excellent survey of clone theory.

The structure given by Post's Lattice enables us to compare the complexity of our circuit-related problems for the cases in which the corresponding clones are comparable. For circuits, we get a stronger result than Proposition~\ref{prop:trivial reductions}: The complexity of our problems does not depend on the actual set $B$ of Boolean functions, but just on the clone $\clone B$ generated by it. Again, an analogous result holds for the tautology problem.

\begin{lemma}\label{lemma:classes work multi-modal}
  Let $B_1,B_2$ be finite sets of Boolean functions, $\mathcal F$ a class of frames, $k\ge 1$, and $M\subseteq\set{\Diamond,\Box}$. If $B_1\subseteq\clone {B_2},$ then $\circsat{\mathcal F}{M}{k}{B_1}\redlogm\circsat{\mathcal F}{M}{k}{B_2}.$
\end{lemma}

\begin{proof}
This reduction is achieved by replacing every occurring gate representing a function from $B_1$ with the appropriate $B_2$-circuit computing the same function. The resulting circuit obviously is $\mathcal F$-equivalent to the original circuit.
\end{proof}

It is worth noting that an analogous result for formulas cannot be obtained in such an easy way, as the following example illustrates: Consider the sets $B_1=\set{\oplus}$ and $B_2=\set{\wedge,\vee,\neg}$ of Boolean functions. Since every Boolean function can be represented using only AND, OR, and negation gates, it is obvious that $B_1\subseteq\clone{B_2}$ holds. However, a reduction from $\formsat{\K}{\emptyset}{0}{B_1}$ to $\formsat{\K}{\emptyset}{0}{B_2}$ cannot be achieved in a straightforward manner, as a formula transformation analogous to the proof of Lemma~\ref{lemma:classes work multi-modal} would replace a subformula $\varphi_1\oplus\varphi_2$ with the formula $\left(\varphi_1\wedge\neg\varphi_2\right)\vee\left(\neg\varphi_1\wedge\varphi_2\right),$ and repeated application of this transformation leads to exponential size for nested formulas. However, we will see that in the cases arising in this paper, the complexity of a problem $\formsat{\mathcal F}{M}{k}{B}$ also only depends on the clone generated by $B.$

\section{The Satisfiability Problem}\label{section:satisfiability}

Our main results are the classification theorems which we will present now. A graphical presentation of these results can be found in Figures~\ref{fig:k and both modal operators} and~\ref{fig:kd and k with one modal operator}.
For the most general problem of \K-satisfiability, we get the following trichotomy:

\begin{figure}
  \hspace*{-0.525cm}\resizebox{150mm}{190mm}{
  \includegraphics{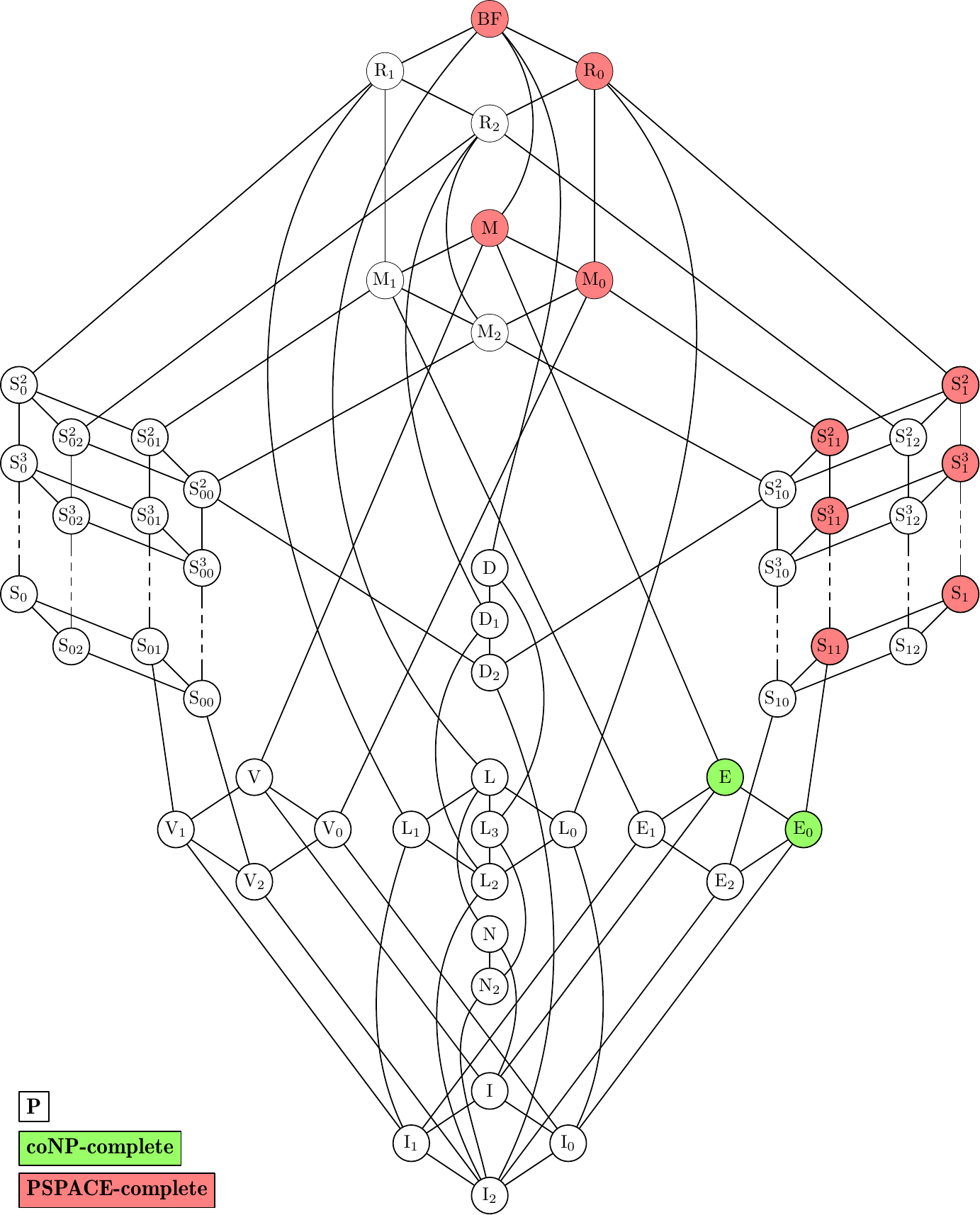}}
  \caption{The complexity of $\formsat{\K}{\Box,\Diamond}{k}{B}$ for $k\ge 1$ and $\circsat{\K}{\Box,\Diamond}{k}{B}$.}
  \label{fig:k and both modal operators}
\end{figure}

\begin{figure}
  \hspace*{-0.525cm}\resizebox{150mm}{190mm}{
  \includegraphics{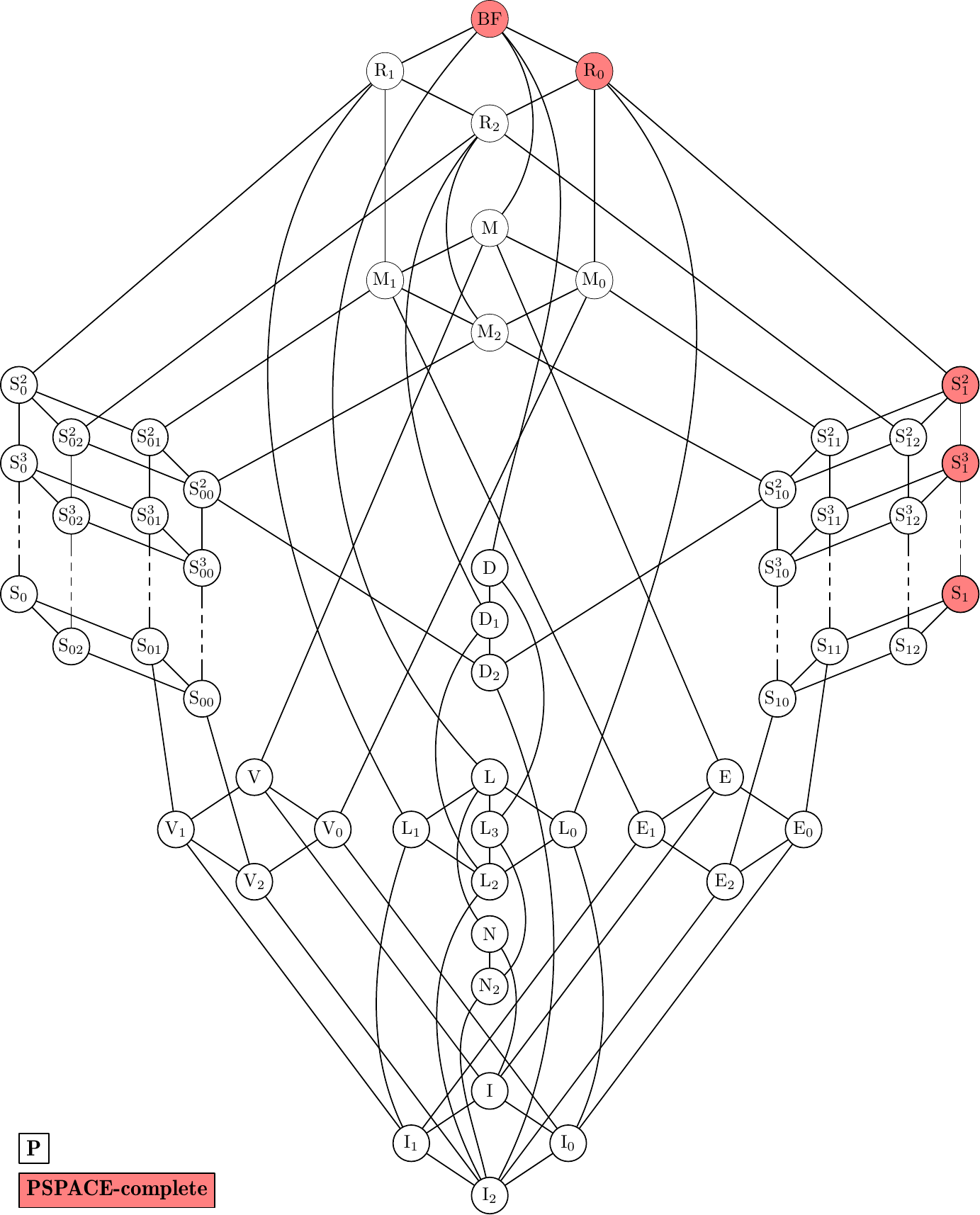}}
  \caption{The complexity of 
  $\formsat{\KD}{M}{k}{B}$ and $\circsat{\KD}{M}{k}{B}$ for any $\emptyset\neq M\subseteq\set{\Box,\Diamond}$ and $\formsat{\K}{\Diamond}{k}{B}$, $\formsat{\K}{\Box}{k}{B}$, $\circsat{\K}{\Diamond}{k}{B}$, and $\circsat{\K}{\Box}{k}{B}$.}
  \label{fig:kd and k with one modal operator}
\end{figure}

\begin{theorem} 
Let $B$ be a finite set of Boolean functions, $k\ge 1$, and $\emptyset\neq M\subseteq\set{\Box,\Diamond}$. Then the following holds:
\begin{itemize}
 \item If $B\subseteq \cR_1,\cD,\cV,$ or $\cL,$ then $\formsat{\K}{M}{k}{B},\circsat{\K}{M}{k}{B}\in\PTIME$ (Corollary~\ref{corollary:r1 and d satisfiable and in p}, Theorem~\ref{theorem:subsets on v are in p}, Theorem~\ref{theorem:xor in p for k and kd}).
 \item If $\cE_0\subseteq\clone B\subseteq\cE,$ then $\formsat{\K}{M}{k}{B},\circsat{\K}{M}{k}{B}\in\PTIME$ if $\card M\leq 1$, and are \CONP-complete otherwise (Section~\ref{sect:conp completeness}, Theorem~\ref{theorem:monotone in p for one modal operator and k or k4}).
\item if $\cS_{11}\subseteq\clone B\subseteq \cM,$ and $\formsat{\K}{M}{k}{B}$ and $\circsat{\K}{M}{k}{B}$ are \PSPACE-complete if $M=\set{\Box,\Diamond}$, and in $\PTIME$ otherwise (Corollary~\ref{cor:s11 formulas pspace complete}, Theorem~\ref{theorem:monotone in p for one modal operator and k or k4}).
\item Otherwise, $\cS_1\subseteq\clone B$ and $\formsat{\K}{M}{k}{B}$ and $\circsat{\K}{M}{k}{B}\in\PTIME$ are \PSPACE-complete (Corollary~\ref{cor:s11 formulas pspace complete}).
\end{itemize}
\end{theorem}

For the logic \KD, we get the following complete classification:

\begin{theorem}
Let $B$ be a finite set of Boolean functions, $k\ge 1$, and $\emptyset\neq M\subseteq\set{\Box,\Diamond}$. Then the following holds:
\begin{itemize}
 \item If $B\subseteq \cR_1,\cD,\cM,$ or $\cL,$ then $\formsat{\KD}{M}{k}{B},\circsat{\KD}{M}{k}{B}\in\PTIME$ (Corollary~\ref{corollary:r1 and d satisfiable and in p}, Theorem~\ref{theorem:monotone formulas in p for classes below kd}, Theorem~\ref{theorem:xor in p for k and kd}).
 \item Otherwise, $\cS_{1}\subseteq\clone B,$ and $\formsat{\KD}{M}{k}{B}$ and $\circsat{\KD}{M}{k}{B}$ are \PSPACE-complete (Corollary~\ref{cor:s1 formulas pspace complete}).
\end{itemize}
\end{theorem}

This dichotomy is a natural analog of Lewis's result that the satisfiability problem for Boolean formulas with connectives from $B$
is \NP-complete if $\mathtext{S}_1\subseteq\clone B$ and in \PTIME\ otherwise \cite{lew79}.

From these theorems, we conclude that using the more succinct representation of modal circuits does not increase the polynomial degree of the complexity of these satisfiability problems (for two problems $A$ and $B$, we write $A\redeqpm B$ if $A\redpm B$ and $B\redpm A$).

\begin{corollary}
Let $B$ be a finite set of Boolean functions, $\mathcal F\in\set{\K,\KD}$, $k\ge 1$, and let $M\subseteq\set{\Box,\Diamond}$. Then
$\circsat{\mathcal F}{M}{k}{B} \redeqpm \formsat{\mathcal F}{M}{k}{B}.$
\end{corollary}

The following is our classification for the logics $\T$ and $\Sfour,$ which gives a complete classification except for the cases where $\clone B$ is one of the clones $\cL$ or $\cL_0.$

\begin{theorem}
 Let $B$ be a finite set of Boolean functions, $\mathcal F\in\set{\T,\Sfour}$, $k\ge 1$, and $\emptyset\neq M\subseteq\set{\Box,\Diamond}$.

\begin{itemize}
 \item If $B\subseteq \cR_1,\cD,\cN$ or $\cM,$ then $\circsat{\mathcal F}{M}{k}{B}\in\PTIME$ (Corollary~\ref{corollary:r1 and d satisfiable and in p}, Theorem~\ref{theorem:monotone formulas in p for classes below kd}, Theorem~\ref{theorem:negation in polynomial time})
 \item If $\cS_1\subseteq\clone B,$ then $\circsat{\mathcal F}{M}{k}{B}$ is \PSPACE-complete.
\item Otherwise, $\clone B\in\set{\cL,\cL_0}.$
\end{itemize}
\end{theorem}

The logic \Sfive\ behaves differently: It is well known that the satisfiability problem for this logic can be solved in \NP, as long as only one modality is present \cite{lad77}. As soon as at least two modalities are involved, the problem becomes \PSPACE-complete \cite{hamo92}. We show that, in a similar way to the other logics with \PSPACE-complete satisfiability problems that we considered, the problem is hard for this complexity class as soon as the propositional functions we allow in the formulas and circuits can express the crucial function $x\wedge\overline y$, which corresponds to clones that are supersets of $\cS_1$.

\begin{theorem}
 Let $B$ be a finite set of Boolean functions, $k\ge1$, and $\emptyset\neq M\subseteq\set{\Box,\Diamond}$. Then the following holds:
\begin{itemize}
 \item If $B\subseteq \cR_1,\cD,\cN$ or $\cM,$ then $\circsat{\Sfive}{M}{k}{B}\in\PTIME$ (Corollary~\ref{corollary:r1 and d satisfiable and in p}, Theorem~\ref{theorem:monotone formulas in p for classes below kd}, Theorem~\ref{theorem:negation in polynomial time})
 \item If $\cS_1\subseteq\clone B,$ then $\circsat{\Sfive}{M}{k}{B}$ is \PSPACE-complete if $k\ge 2$, and \NP-complete if $k=1$ (Corollary~\ref{cor:s1 formulas pspace complete}).
\item Otherwise, $\clone B\in\set{\cL,\cL_0}.$
\end{itemize}
\end{theorem}

The above classifications leave open the cases where the set $B$ generates one of the clones $\cL$ and $\cL_0$. We will discuss these open issues in Section~\ref{sec:P}. Note that in the above theorem, the \NP-hardness results are immediate from the previously mentioned results in \cite{lew79}: It directly follows from his result that for any non-empty class $\mathcal F$ of frames, the problem $\formsat{\mathcal F}{\emptyset}{0}{B}$ is \NP-hard if $\cS_1\subseteq\clone B$.

The rest of this section is devoted to proving these theorems. As mentioned before, if suffices to prove upper bounds for circuits and lower bounds for formulas.

\subsection{General Upper Bounds}

It is well known that the ${\cal F}$-satisfiability problem for modal
formulas using the operators $\Box,\wedge$, and $\neg$ is
solvable in \PSPACE\ for a variety of classes $\mathcal F$ of frames for both the uni-modal case
\cite{lad77} and the general multi-modal setting \cite{hamo92}. The following theorem shows that the circuit case can be reduced to
the formula case, thus putting the circuit problems in \PSPACE\ as well.

The intuitive reason why the complexity of our satisfiability problems does not
increase significantly when considering circuits instead of formulas is that for many
algorithms in modal logic, the complexity depends on the number of appearing
subformulas more than on the length of the formula.

\begin{theorem}
  \label{th:circuit-to-formula}
  Let $B$ be a finite set of Boolean functions, $\cal F\in\{\K,\KD,\T,\Sfour,\Sfive\}$, $k\ge1$, and $M\subseteq\set{\Box,\Diamond}$. Then $\circsat{\mathcal F}{M}{k}{B}\in\PSPACE$ and $\circsat{\Sfive}{M}{1}{B}\in\NP$.
\end{theorem}

\begin{proof}
 The main idea of the proof is to transform the given circuit in $\multcirc MkB$ into a modal formula using modal operators $\Box_1,\dots,\Box_k$, the modal operator $E$ (where $E\varphi$ is an abbreviation for $\Box_1\varphi\wedge\dots\wedge\Box_k\varphi$), and the propositional symbols $\wedge,\vee,\neg$. Satisfiability for these formulas for the classes $\mathcal F$ of frames that we consider can be solved in \PSPACE\ and the case where $\mathcal F=\mathtext{S}5$ and $k=1$ can be solved in \NP\ \cite{lad77,hamo92}. Note that their proofs do not cover the $E$-operator, but they work without any change if $E\varphi$ is always locally evaluated as its expansion $\Box_1\varphi\wedge\dots\wedge\Box_k\varphi$ in the algorithms presented in~\cite{hamo92}.

The reduction works as follows: Let $C$ be a circuit in $\multcirc MkB$ modal $B$-circuit with up to $k$ modalities. Due to Lemma~\ref{lemma:classes work multi-modal} and since $\PSPACE$ is closed under $\redlogm$-reductions, we can without loss of generality assume that $B=\set{\wedge,\neg}$. For every gate $g$ in $C$, define
  $f'(C,g)$ as follows:

\begin{itemize}
\item If $g$ is an input gate labeled $x_i$, then $f'(C,g) = g \leftrightarrow
  x_i$.
\item If $g$ is a $\neg$-gate, then $f'(C,g) = g \leftrightarrow \neg
  h$, where $h$ is the predecessor gate of $g$ in $C$.
\item If $g$ is an $\wedge$-gate, then $f'(C,g) = g \leftrightarrow
  (h_1 \wedge h_2)$, where $h_1, h_2$ are the predecessor gates of $g$
  in $C$.
\item If $g$ is a $\Box_i$-gate for some $1\leq i\leq k$, then $f'(C,g) = g \leftrightarrow \Box_i
  h$, where $h$ is the predecessor gate of $g$ in $C$.
\end{itemize}

  In this way, the gates of the circuit are represented by variables in the
  corresponding formula.  We will view $f'(C,g)$ as a formula over
  $\{\Box_1,\dots\Box_k, \wedge, \neg\}$, by viewing  ``$\varphi \leftrightarrow \psi$'' as 
  shorthand for
  ``$\neg (\varphi \wedge \neg \psi) \wedge \neg (\neg \varphi \wedge \psi)$.''
  Clearly, $f'$ is computable in logarithmic
  space (note that the $\leftrightarrow$ symbols do not occur nested).
  We now define the actual reduction as follows: For every
  circuit $C\in\multcirc Mk{\wedge,\neg}$ with output gate $g_{\mathtext{out}}$,
  \[f(C) = g_{\mathtext{out}} \wedge \bigwedge_{g \mbox{\scriptsize{ gate in
      }} C}\bigwedge_{i=0}^{\smd C} E^i f'(C,g).\] Here $E^i \varphi$ denotes $\underbrace{E\dots E}_{i\mathtext{ times}}\varphi$. Clearly, $f$ is
  computable in logarithmic space.
We will now show that $C$ is ${\cal F}_k$-satisfiable if and
only if $f(C)$ is ${\cal F}_k$-satisfiable.

First suppose that $C$ is ${\cal F}_k$-satisfiable. Let $M=(W,R_1,\dots,R_k,\pi)$ be an $\mathcal F_k$-model, and let $w_0 \in W$ be a world such that $M,w_0 \models C$. The model $M'$ is defined over the same set of worlds with the same successor relations, and inherits the truth assignment from $M$ for all variables appearing in $C$. For the new variables, the truth assignment $\pi'$ of $M'$ is defined as follows: For every
gate $g$ in $C$, $\pi'(g) = \{w \in W \ | \ M,w \models C_g\}$. Here
$C_g$ is the subcircuit of $C$ with output gate $g$. By definition of
$\pi'$, for every world $w \in W$ and for every gate $g \in C$, $M',w
\models g$ if and only if $M',w \models C_g$. It is easy to show
(see below) that for every world $w \in W$ and for every gate
$g\in C$, $M',w \models f'(C,g)$. This implies that $M',w_0 \models\bigwedge_{g
\mbox{\scriptsize{ gate in }} C}\bigwedge_{i=0}^{\smd C} E^i
f'(C,g)$. Since $M,w_0 \models C$ and $C = C_{g_{\mathtext{out}}}$,
it follows by the definition of $\pi'$ that $M',w_0 \models g_{\mathtext{out}}$.
It follows that $M',w_0 \models f(C)$, and thus $f(C)$ is ${\cal F}$-satisfiable.

To be complete, we will show that, as mentioned above, for every world $w \in W$ and for
every gate $g \in C$, $M',w \models f'(C,g)$. We make a case distinction.

\begin{itemize}
\item $g$ is an input gate $x_i$. By definition of $\pi'$, $M',w
  \models g$ if and only if $M',w \models x_i$. It follows that $M',w
  \models g \leftrightarrow x_i$.

\item $g$ is a $\neg$-gate.  Let $h$ be the predecessor gate of $g$.
  $M',w \models g$ if and only if $M',w \models C_g$. The latter holds
  if and only if $M',w \not \models C_h$. This holds if and only if
  $M',w \not \models h$. It follows that $M',w \models g
  \leftrightarrow \neg h$.
 
\item $g$ is an $\wedge$-gate.  Let $h_1$ and $h_2$ be the predecessor
  gates of $g$. $M',w \models g$ if and only if $M',w \models C_g$.
  The latter holds if and only if $M',w \models C_{h_1}$ and $M',w
  \models C_{h_2}$.  By definition of $\pi'$, $M',w \models C_{h_1}$
  if and only if $M',w \models h_1$ and $M',w \models C_{h_2}$ if and
  only if $M',w \models h_2$. It follows that $M',w \models g
  \leftrightarrow (h_1 \wedge h_2)$.

\item $g$ is a $\Box_i$-gate for some $i$.  Let $h$ be the predecessor gate of $g$.
  $M',w \models g$ if and only if $M',w \models C_g$. The latter holds
  if and only if $(\forall w' \in W)[wR_i w' \Rightarrow M',w' \models
  C_h]$. This holds if and only if $(\forall w' \in W)[wR_i w'
  \Rightarrow M',w' \models h]$. It follows that $M',w \models g
  \leftrightarrow \Box_i h$.
\end{itemize}

For the converse, suppose that $f(C)$ is ${\cal F}$-satisfiable. Let
$M$ be an ${\cal F}$-model, and let
$w_0 \in W$ be a world such that $M,w_0 \models f(C)$. We will prove
by induction on the structure of circuit $C_g$ that for every gate $g
\in C$ and for every world $w$ that is reachable from $w_0$ in at most
$\md C - \md{C_g}$ steps, $M,w \models C_g$ if and only if $M,w
\models g$. This clearly implies that $M,w_0 \models C$, and thus $C$
is ${\cal F}$-satisfiable.

\begin{itemize}
\item $g$ is an input gate $x_i$. Then $C_g$ is equivalent to $x_i$.
  Since $M,w \models g \leftrightarrow x_i$, it follows that $M,w
  \models C_g$ if and only if $M,w \models g$.
\item $g$ is a $\neg$-gate.  Let $h$ be the predecessor gate of $g$.
  Then $M,w \models C_g$ if and only if $M,w \not \models C_h$. By
  induction, the latter holds if and only if $M,w \not \models h$.
  Clearly, $M,w \not \models h$ if and only if $M,w \models \neg h$.
  Since $M,w \models g \leftrightarrow \neg h$, it follows that $M,w
  \models C_g$ if and only if $M,w \models g$, as required.

\item $g$ is an $\wedge$-gate.  Let $h_1$ and $h_2$ be the predecessor
  gates of $g$. Then $M,w \models C_g$ if and only if $M,w \models
  C_{h_1}$ and $M,w \models C_{h_2}$. By induction, the latter holds
  if and only if $M,w \models h_1$ and $M,w \models h_2$, and this
  holds if and only if $M,w \models h_1 \wedge h_2$. Since $M,w
  \models g \leftrightarrow (h_1 \wedge h_2)$, it follows that $M,w
  \models C_g$ if and only if $M,w \models g$, as required.

\item $g$ is a $\Box_i$-gate for some $i$. Let $h$ be the predecessor gate of $g$.
  Then $M,w \models C_g$ if and only if for all $w' \in W$ such that
  $wR_i w'$, it holds that $M,w' \models C_h$. Note that $\md{C_h} =
  \md{C_g} - 1$. Since $w$ is reachable from $w_0$ in at most $\md{C}
  - \md{C_g}$ steps, it follows that for every $w'$ such that $wR_i w'$, $w'$ is
  reachable from $w_0$ in at most $\md{C} - \md{C_g} + 1 = \md{C} -
  \md{C_h}$ steps. And so, by induction, it follows that (for all $w'
  \in W$ such that $wR_i w'$, it holds that $M,w' \models C_h$) if and
  only if (for all $w' \in W$ such that $wR_i w'$, it holds that $M,w'
  \models h$), and this holds if and only if $M,w \models \Box_i h$.
  Since $M,w \models g \leftrightarrow \Box_i h$, it follows that $M,w
  \models C_g$ if and only if $M,w \models g$, as required.
\end{itemize}
Finally note that the \KD\ case easily follows from the result for \K, since
a circuit $C$ is \KD-satisfiable if and only if
$C\wedge\bigwedge_{i=0}^{\smd{\varphi}}E^i\bigwedge_{j=1}^k\Diamond_j 1$ is \K-satisfiable.
\end{proof}

Note that in the uni-modal case, we do not have to introduce the $E$-operator as in the proof above. Therefore the construction of the proof directly implies that for any class $\mathcal F$ of frames, uni-modal satisfiability for circuits (using any set of propositional gates) is not more difficult than the satisfiability problem for $\set{\wedge,\neg}$-formulas for the same class of frames.

\begin{corollary}\label{corollary:circuit-to-formula}
  Let $B$ be a finite set of Boolean functions and $\mathcal F$ a class of frames. Then $\circsat{\mathcal F}{\Box,\Diamond}{1}{B}\redpm\formsat{\mathcal F}{\Box}{1}{\wedge,\neg}$.
\end{corollary}

\subsection{\PSPACE-completeness}

We now show how to express, in a satisfiability-preserving way, uni-modal formulas and circuits using a restricted set of Boolean connectives and one modal operator. This implies that our satisfiability problems for these restricted sets of formulas are as hard as the general case.

As mentioned in the discussion following Lemma~\ref{lemma:classes work multi-modal}, with many formula transformations, the size of the resulting formula can be exponential. A crucial tool in dealing with this situation is the following lemma showing that for certain sets $B,$ there are always short formulas representing the functions AND, OR, and NOT.  Part (1) is Lemma~1.4.5 from \cite{hschnoor07:phdthesis}, the result for the case $\clone B=\cBF$ is proven in \cite{lew79}. Part (2) follows directly from the proofs in \cite{lew79}.

\begin{lemma}\label{lemma:or formula}\label{lemma:lewis complete and
    negation} Let $B$ be a finite set of Boolean functions.
  \begin{enumerate}
  \item If $V\subseteq\clone{B}$ ($E\subseteq\clone{B}$, resp.), then there exists a $B$-formula $f(x,y)$ such that $f$ represents $x\vee y$ ($x\wedge y$, resp.) and each of the variables $x$ and $y$ occurs exactly once in $f(x,y)$.
  \item If $\mathtext{N}\subseteq\clone{B}$, then there exists a $B$-formula $f(x)$ such that $f$ represents $\overline x$ and the variable $x$ occurs in $f$ only once.
  \end{enumerate}
\end{lemma}

The proof of the following theorem uses a generalization of ideas from the proof for the main result in \cite{lew79}. This can be applied to an arbitrary class of frames, and in particular, it yields \PSPACE\ completeness results for \K\ and \KD.

\begin{theorem}\label{theorem:s1 implementation}
  Let $B$ be a finite set of Boolean functions such that $\mathtext{S}_1\subseteq\clone B$, ${\mathcal F}$ a class of frames, and $\emptyset\neq M\subseteq\{\Box,\Diamond\}.$ 
  Then the following holds: 
  \begin{itemize}
  \item $\formsat{\mathcal F}{\Box, \Diamond}{1}{\wedge,\neg}\redlogm\formsat{\mathcal F}{M}{1}{B},$ 
  \item  $\formsat{\Sfive}{\Box,\Diamond}{2}{\wedge,\neg}\redlogm\formsat{\Sfive}{M}{2}{B}$.
  \end{itemize}
\end{theorem}

\begin{proof} 
First consider the uni-modal case. Let $\varphi\in\multform{\Box,\Diamond}{1}{\wedge, \neg}.$ Without loss of generality, assume that $\varphi$ contains only modal operators from $M$ (use the identity $\Box\equiv\neg\Diamond\neg$ otherwise). Let $B':=B\cup\{1\}$. Then Figure~\ref{fig:lattice} shows that $\clone{B'}=\rm BF$ (since ${\rm I}_1$ is the smallest clone containing 1, and {\rm BF} is the smallest clone containing ${\rm I}_1$ and ${\rm S}_1$). It follows from Lemma~\ref{lemma:lewis complete and negation} that there is a $B'$-formula $f_{\neg}(x)$ that represents $\overline x$, and $x$ occurs in $f_{\neg}(x)$ only once, and there exist $B'$-formulas $f_\wedge(x,y)$ and $f_\vee(x,y)$ such that $f_\wedge$ represents $\wedge$, $f_\vee(x,y)$ represents $\vee$, and $x$ and $y$ occur exactly once in $f_\wedge(x,y)$ and exactly once in $f_\vee(x,y)$. In $\varphi$, replace every occurrence of $\wedge$ with $f_\wedge,$ every occurrence of $\vee$ with $f_\vee$, and every occurrence of $\neg$ with $f_\neg.$ Call the resulting formula $\varphi'$. Clearly, $\varphi'$ is a formula in $\multform{M}{1}{B'},$ and  $\varphi'$ is equivalent to $\varphi.$ By choice of $f_\vee$, $f_\wedge$, and $f_\neg$, $\varphi'$ is computable in polynomial time.

Now replace every occurrence of the constant $1$ with a new variable $t$ and force $t$ to be $1$ in every relevant world by adding $\wedge\bigwedge_{i=0}^{\smd \varphi}\Box_1^i t.$ This is a conjunction of linearly many terms (since $\md{\varphi}\leq\card{\varphi}$). We insert parentheses in such a way that we get a tree of $\wedge$'s of logarithmic depth. Now express the $\wedge$'s in this tree with the  equivalent $B$-formula (which exists, since $\clone B\supseteq {\rm S}_{1}\supset{\rm E}_2 = \clone{\wedge}$) with the result only increasing polynomially in size. It is obvious that this formula is satisfiable if and only if the original formula is.

Now for the bimodal case and the logic $\Sfive,$ we use the same construction as above, except that to force the variable $t$ to true in all relevant worlds, we use the formula $(\Box_1\Box_2)^{\smd\varphi}t.$ Due to the reflexivity of both successor relations in $\Sfive_2$-models, this forces $t$ to be true in all relevant worlds.
\end{proof}

The following theorem implies that for the logic \K, \PSPACE-completeness already holds for a lower class in Post's Lattice. The proof is nearly identical to the one for the above Theorem~\ref{theorem:s1 implementation}: Note that $\clone{\mathtext{S}_{11}\cup \{1\}}=\mathtext{M}$, and apply Lemma~\ref{lemma:lewis complete and negation} for the class $\mathtext{M}$. Then follow the construction above. (We can represent $\wedge$ by a $B$-formula since ${\rm S}_{11}\supseteq{\rm E}_2=\clone\wedge$, and we can represent $0$ by a $B$-formula since $0\in{\rm S}_{11}$.)

\begin{theorem}\label{theorem:s11 implementation}
  Let $B$ be a finite set of Boolean functions such that $\mathtext{S}_{11}\subseteq\clone{B}$, ${\mathcal F}$ a class of frames, $k\ge 1$, and $M\subseteq\{\Box,\Diamond\}$. Then $\formsat{\mathcal F}{M}{1}{\wedge,\vee, 0}\redlogm\formsat{\mathcal F}{M}{1}{B}$.
\end{theorem}

The above theorems give the following corollary.

\begin{corollary}\label{cor:s1 formulas pspace complete}\label{cor:s11 formulas pspace complete}
  Let $B$ be a finite set of Boolean functions, and let $\emptyset\neq M\subseteq\set{\Box,\Diamond}$.
  \begin{enumerate}
     \item If $\clone B\supseteq\cS_1$, and $\mathcal F$ is a class of frames such that $\Sfour\subseteq\mathcal{F}\subseteq\K$, and $k\ge 1$, then $\formsat{\mathcal F}{M}{k}{B}$ and $\circsat{\mathcal F}{M}{k}{B}$ are \PSPACE-hard.
     \item If $\clone B\supseteq\cS_{11}$ and $k\ge 1$, then $\formsat{\K}{\Box,\Diamond}{k}{B}$ and $\circsat{\K}{\Box,\Diamond}{k}{B}$ are \PSPACE-complete.
     \item If $\clone B\supseteq\cS_1$ and $k\ge 2$, then $\formsat{\Sfive}{M}{k}{B}$ and $\circsat{\Sfive}{M}{k}{B}$ are \PSPACE-complete.
  \end{enumerate}
\end{corollary}

\begin{proof}
  The upper bounds follow from Theorem~\ref{th:circuit-to-formula}.
  \begin{enumerate}
  \item {In \cite{lad77}, it is shown that for every class of frames $\mathcal F$
         such that $\Sfour\subseteq\mathcal{F}\subseteq\K,$
         the problem $\formsat{\mathcal F}{M}{1}{\wedge,\neg}$ is \PSPACE-hard. Therefore this follows from \cite{lad77} and Theorem~\ref{theorem:s1 implementation}.}
  \item {In~\cite[Theorem 6.5]{hem01}, it is shown that
      $\formsat{\K}{\Box,\Diamond}{1}{\wedge,\vee,0}$ is \PSPACE-hard. Thus
      the result follows from Theorem~\ref{theorem:s11
        implementation}.}
  \item {In \cite{hamo92}, it is shown that $\formsat{\mathtext{S}5}{\Box,\Diamond}{2}{\wedge,\neg}$ is \PSPACE-hard. Therefore, the result follows from Theorem~\ref{theorem:s1 implementation}.}
  \end{enumerate}
\end{proof}

\subsection{\CONP-completeness}\label{sect:conp completeness}

In \cite{hem01}, the analogous result of the following lemma was shown
for uni-modal formulas. We prove that this \CONP\ upper bound also holds for circuits.

\begin{lemma}\label{lemma:e in conp}
  Let $k\geq 1$. Then $\circsat{{\rm K}}{\Box,\Diamond}{k}{\wedge,0,1}\in\CONP.$
\end{lemma}

\begin{proof}
The proof for the analogous statement for uni-modal formulas is based on 
the following fact: Let $\varphi$ be a formula of the form 
$\varphi=\bigwedge_{i\in I}\Box\varphi^{\Box}_i \wedge
      \bigwedge_{j\in J}\Diamond\varphi^{\Diamond}_j \wedge
      \psi$, where $I$ and $J$ are finite sets of 
indices, $\varphi^{\Box}_i$ and $\varphi^{\Diamond}_j$ are modal formulas 
for all $i\in I$, $j\in J$, and $\psi$ is a propositional formula.
Then $\varphi$ is satisfiable if and only if $\psi$ is satisfiable and for 
every $j\in J$,
$\bigwedge_{i\in I}\varphi^{\Box}_i \wedge\varphi^{\Diamond}_j$ is satisfiable  \cite{lad77}.

This generalizes to multi-modal formulas from $\multform{\Box,\Diamond}{k}{\wedge,0,1}$ in the following way: let $$\varphi=\bigwedge_{i\in I_1}\Box_1\varphi^{\Box_1}_i \wedge \dots \wedge \bigwedge_{i\in I_k}\Box_k\varphi^{\Box_k}_i\wedge
      \bigwedge_{j\in J_1}\Diamond_1\varphi^{\Diamond_1}_j \wedge\dots\wedge\bigwedge_{j\in J_k}\Diamond_k\varphi^{\Diamond_k}_j\wedge
      \psi,$$ 

for finite sets of indices $I_1,\dots,I_k,J_1,\dots,J_k$, formulas $\varphi^{\Box_l}_i,\varphi^{\Diamond_l}_j\in\multform{\Box,\Diamond}{k}{\wedge,0,1}$, and a propositional $\set{\wedge,0,1}$-formula $\psi$. Then $\varphi$ is satisfiable if and only if for every $1\leq l\leq k$ and
every $j\in J_l$ it holds that $\psi$ and
$\bigwedge_{i\in I_l}\varphi^{\Box_l}_i \wedge\varphi^{\Diamond_l}_j$ are satisfiable. Since every formula from $\multform{\Box,\Diamond}{k}{\wedge,0,1}$ can be written in the above form and since
satisfiability for the propositional part $\psi$ can be tested in polynomial time according to \cite{lew79}, this leads to a recursive
$\NP$-algorithm for the question if $\varphi$ is unsatisfiable.

We give an analogous proof for multi-modal circuits. Let $C$ be a circuit from
$\multcirc{\Box,\Diamond}{k}{\wedge,0,1}$ with output-gate $\mathit{out}$. If $\mathit{out}$ is a $\Box_i$-gate for some $1\leq i\leq k$, 
then $\varphi$ is satisfied in every world without a successor, if  $\mathit{out}$ is a 
$\Diamond_i$-gate for some $1\leq i \leq k$, then $C$ is satisfiable if and only if the circuit obtained from $C$ by using 
the predecessor of $\mathit{out}$ as output-gate is satisfiable, and finally if $\mathit{out}$ is an 
input-gate or a constant gate, then satisfiability can be tested
trivially. Therefore we assume without 
loss of generality $\mathit{out}$ to be an $\wedge$-gate. For a 
set of gates $G$ we define $\pred(G)$ to be the set of all direct predecessor gates of 
gates in $G$ and $\wpred(G)$ to be the set of all non $\wedge$-gates $g$ which 
are connected to $G$ by a path from $g$ to a gate $g'\in G$ where all gates on 
the path excluding $g$ (but including $g'$ if $g\neq g'$)
are $\wedge$-gates. 

For $1\leq i\leq k$ let $G_{\Box_i}$ be the set of all 
$\Box_i$-gates in $C$, $G_{\Diamond_i}$ the set of all $\Diamond_i$-gates in $C$ and 
$G$ the set of all propositional gates in $C$. Then, due to the equivalence above,  
$C$ is satisfiable if and only if 
$$\bigwedge_{g\in \wpred(\{\mathit{out}\}) \cap G} \! \!\!\!\! \!\varphi_g  \ \ \ 
\text{ and }\  \bigwedge_{g\in \pred(\wpred(\{\mathit{out}\}) \cap G_{\Box_i})} 
\!\!\! \!\!\!\!\!\varphi_g \wedge \bigwedge_{g \in \pred(\{g_{\Diamond_i}\})} \! \! 
\varphi_g$$ are satisfiable for every $1\leq i\leq k$ and every $g_{\Diamond_i}\in \wpred(\{\mathit{out}\}) \cap 
G_{\Diamond_i}$, where for a gate $g$, the formula $\varphi_g$ is defined as in the 
definition for modal circuits, i.e., $\varphi_g$ is the formula represented by the 
sub-circuit with output-gate $g$. Note that due to the definition of $\wpred$, the first of these formulas is a propositional formula.

More generally, a formula of the form 
$ \varphi= \bigwedge_{g\in H} \varphi_g$ for a set $H$ of gates from $C$ is 
satisfiable if and only if 
$$\psi:=\!\!\!\bigwedge_{g\in \wpred(H) \cap G} \! 
\!\!\!\! \!\varphi_g  \ \ \ \text{ and }\  
\varphi^{g_{\Diamond_i}}:=\!\!\!\bigwedge_{g\in 
\pred(\wpred(H) \cap G_{\Box_i})} \!\!\! \!\!\!\!\!\varphi_g \wedge 
\bigwedge_{g \in \pred(\{g_{\Diamond_i}\})} \! \! \varphi_g$$ are satisfiable for every $1\leq i\leq k$ and every 
$g_{\Diamond_i}\in \wpred(H) \cap G_{\Diamond_i}$.

Note that $\psi$ is a conjunction of constants and variables, therefore satisfiability of $\psi$ can be tested in polynomial time. It is obvious that constructing the sets $\pred(H)$ and 
$\wpred(H)$ needs only polynomial time as well.

For testing if a formula $\varphi$ represented by $H$ is unsatisfiable it suffices to check if $\psi$ is unsatisfiable, and, if this is not the case, to guess a 
$g_{\Diamond_i}\in \wpred(H) \cap G_{\Diamond_i}$ for some $1\leq i\leq k$ and to recursively test 
unsatisfiability of $\varphi^{g_{\Diamond_i}}$, which is represented by the set 
$\pred(\wpred(H) \cap G_{\Box_i})\cup\pred(\{g_{\Diamond_i}\})$. Since in every 
recursion the length of the longest path between an input-gate and a gate in $H$ 
decreases, the algorithm stops after at most $|C|$ recursions.

Hence, starting with $H=\{\mathit{out}\}$ we get an NP-algorithm for testing 
unsatisfiability of $C$.
\end{proof} 

In~\cite{hem05}, it is shown that $\formsat{\K}{\Box,\Diamond}{1}{\wedge,0}$ is \CONP-hard. Applying Lemma~\ref{lemma:or formula}, we obtain the following result.

\begin{lemma}\label{lemma:e0 conp hard all bases}
Let $B$ be a finite set of Boolean functions such that $\cE\supseteq \clone B\supseteq\cE_0$, and $k\ge 1$. Then $\formsat{\K}{\Box,\Diamond}{k}{B}$ is \CONP-hard.
\end{lemma}

\begin{proof}
 It obviously suffices to consider the case $k=1.$ We use a similar construction as in the proof for Theorem~\ref{theorem:s1 implementation}. Let $B':=B\cup\{1\}.$ From the structure of Post's Lattice, it follows that $\clone {B'}=\cE.$ Hence, by Lemma \ref{lemma:or formula}, we have a short $B'$-formula for AND, and can convert
  $\multform{\Box,\Diamond}{1}{\wedge,0}$-formulas into equivalent formulas from $\multform{\Box,\Diamond}{1}{B'}.$ We remove the occurrences of $1$ as in Theorem~\ref{theorem:s1 implementation}: Introduce a variable $t$ and force it to be $1$ with the logarithmic tree construction. The \CONP-hardness then follows from the above-mentioned result from \cite{hem05}.
\end{proof}

\subsection{Polynomial Time}\label{sec:P}

We now give our polynomial-time algorithms. We will see that in many of those cases where the restriction of the propositional operators to a certain set $B$ leads to a polynomial-time decision procedure in the propositional case, the same is true for the corresponding modal problems. One notable exception is the case of monotone formulas: For propositional monotone formulas, satisfiability can easily be tested, since such a formula is satisfiable if and only if it is satisfied by the constant $1$-assignment. For modal satisfiability, we have seen in Corollary~\ref{cor:s11 formulas pspace complete} that the corresponding problem is as hard as the standard satisfiability problem for modal logic. The other exception concerns formulas using only conjunction and constants: As a special case of monotone formulas, satisfiability testing is easy for propositional logic. However, Section~\ref{sect:conp completeness} showed that the problem is \CONP-complete for modal logic.

\begin{lemma}\label{lemma:if varphi is satisfiable then varphi satisfiable in reflexive singleton}
Let $B$ be a finite set of Boolean functions, $k\ge 1$, and $\varphi\in\multform{\Box,\Diamond}{k}{B}$. If the formula $\varphi^\mathtext{id},$ which is obtained by changing every modal operator in $\varphi$ to the identity, is satisfiable, then $\varphi$ is satisfiable in the reflexive singleton.
\end{lemma}

\begin{proof}
Let $I$ be a propositional assignment satisfying $\varphi^{\mathtext{id}}.$ Let $M$ be the model consisting of the reflexive singleton, where each variable is true if and only if it is true in $I$. Since in this model, every modal operator can only refer to the same single world in the model, the operators are equivalent to the identity function, implying the result.
\end{proof}

It is obvious that every propositional $B$-formula for $B\subseteq \cR_1$ or $B\subseteq \cD$ is satisfiable (\cite{lew79}): In the first case, the all-$1$-assignment always is a model. In the second case, exactly one of the two constant assignments is. Hence, Lemma~\ref{lemma:if varphi is satisfiable then varphi satisfiable in reflexive singleton} immediately gives the following complexity result:

\begin{corollary}\label{corollary:r1 and d satisfiable and in p}
Let $B$ be a finite set of Boolean functions such that $B\subseteq\cR_1$ or $B\subseteq\cD$, $\mathcal F$ a class of frames containing the reflexive singleton, and $k\ge 1$. Then every formula from $\multform{\Box,\Diamond}{k}{B}$ is $\mathcal F$-satisfiable. In particular, $\circsat{\mathcal F}{\Box,\Diamond}{k}{B}\in\PTIME$ for $\mathcal F\in\set{\K,\KD,\Kfour,\T,\Sfour,\Sfive}.$
\end{corollary}

While \K-satisfiability for variable-free formulas using constants, the Boolean connectives $\wedge$ and $\vee$, and both modal operators is complete for \PSPACE~\cite{hem01}, this problem (even with variables) is solvable in polynomial time if we look only at frames in which each world has a successor. 

\begin{theorem}\label{theorem:monotone formulas in p for classes below kd}
Let $B$ be a finite set of Boolean functions such that $B\subseteq \cM$, $\mathcal F$ a class of frames such that $\mathcal F\subseteq\KD$, and $k\ge 1$. Then $\circsat{\mathcal F}{\Box,\Diamond}{k}{B}\in\PTIME.$ In particular, $\circsat{\KD}{\Box,\Diamond}{k}{B},\circsat{\T}{\Box,\Diamond}{k}{B},\circsat{\Sfour}{\Box,\Diamond}{k}{B},\circsat{\Sfive}{\Box,\Diamond}{k}{B}\in\PTIME.$
\end{theorem}

\begin{proof}
The claim is obvious if $\mathcal F$ is empty, hence assume that this is not the case. Let $M$ be an $\mathcal F$-model, let $w$ be a world from $M,$ and let $M_1$ be the multi-modal reflexive singleton with $k$ successor relations in which every variable is set to $1.$ It is easy to show by induction on the construction of any $C\in\multcirc{M}{k}{B}$ that if $M,w\models C,$ then $M_1,w\models C$ holds as well. On the other hand, if $M_1,w\models C,$ then $M',w\models C,$ where $M'$ is obtained from the model $M$ by setting every variable to true in every world. Hence, $C$ is $\mathcal F$-satisfiable if and only if $C$ is satisfied in $M_1.$ The latter condition can obviously be verified in polynomial time.
\end{proof}

In the case where all of our propositional operators are unary, we can use simple transformations to decide satisfiability, as the following theorem shows.

\begin{theorem}\label{theorem:negation in polynomial time}
Let $B$ be a finite set of Boolean functions such that $B\subseteq\cN$, $\mathcal F$ a class of frames such that $\mathcal F\in\set{\K,\KD,\Sfour,\Sfive,\Kfour,\T}$, and $k\ge 1$. Then $\circsat{\mathcal F}{\Box,\Diamond}{k}{B}\in\PTIME$.
\end{theorem}

\begin{proof}
Since the clone $\cN$ is generated by negation and the constants, we can, due to Lemma~\ref{lemma:classes work multi-modal}, assume that $B$ only contains these functions.

Now, let $B$ be a circuit from $\multcirc{M}{k}{B}.$ Since every function in $B$ is unary or constant, $C$ is a linear graph, and we can therefore regard $C$ as a formula. Using the equivalence $\Diamond_i \equiv \neg\Box_i\neg,$ we can move negations inward, until we have a formula of the form $O_1\dots O_nz,$ where the $O_i$ are modal operators, and $z$ is either a literal or a constant. It is obvious that this formula is satisfiable if and only if $z$ is not the constant $0,$ or if $\mathcal F=\K,$ and there is at least one $\Box$-operator present. The transformation obviously can be performed in polynomial time.
\end{proof}

For monotone functions and most classes of frames that we are interested in, we already showed that the satisfiability problem can be solved in polynomial time. For the most general class of frames K, this problem is \PSPACE-complete (Corollary~\ref{cor:s11 formulas pspace complete}), but a further restriction of the propositional base gives polynomial-time results here as well.

\begin{theorem}\label{theorem:subsets on v are in p}
Let $B$ be a finite set of Boolean functions such that $B\subseteq\cV$, $\mathcal F$ a class of frames such that $\mathcal F\in\set{\K,\KD,\Sfour,\Sfive,\Kfour,\T}$, and $k\ge 1$. Then $\circsat{\mathcal F}{\Box,\Diamond}{k}{B}\in\PTIME$.
\end{theorem}

\begin{proof}
Since the clone $\cV$ is generated by binary OR and the constants, we can, due to Lemma~\ref{lemma:classes work multi-modal}, assume that $B$ only contains these functions. We first consider the case $\mathcal F\in\set{\K,\Kfour}$.

Let $B$ be a circuit from $\multcirc{\Box,\Diamond}{k}{B}.$ If the output gate $g$ of $C$ is an $\vee$-gate, with predecessors $h_1$ and $h_2$ in $C,$ then $C$ is $\mathcal F$-satisfiable if and only if at least one of $C_{h_1}$ and $C_{h_2}$ is. If $g$ is a $\Diamond_i$-gate with predecessor $h,$ then $C$ is $\mathcal F$-satisfiable if and only if $C_h$ is. Finally, if $g$ is a $\Box_i$-gate, then $C$ is $\mathtext{K}$-satisfiable.

This gives a recursive polynomial-time procedure to decide the satisfiability problem. For the classes other than \K\ and \Kfour, we can use the same procedure, with one exception: here, if $g$ is a $\Box_i$-gate, then $C$ is satisfiable if and only if $C_h$ is satisfiable, where $h$ is the predecessor of $g$ in $C.$
\end{proof}

We now show that for the logics \K\ and \KD, the modal satisfiability problems for formulas having only $\oplus$ and constants in the propositional base are easy. For the propositional case, this holds because unsatisfiable formulas using only these connectives are of a very easy form: Every variable and the constant $1$ appear an even number of times (see, e.g., \cite{lew79}). In the modal case, unsatisfiable formulas over these connectives are of a similarly regular form, as we will soon see. The result also holds for modal circuits.

\begin{theorem}\label{theorem:xor in p for k and kd}
  Let $B$ be a finite set of Boolean functions such that  $B\subseteq\mathtext{L}$, $\mathcal F\in\set{\mathrm{K},\mathrm{KD}}$ a class of frames, and $k\ge 1$. Then $\circsat{\mathcal F}{\Box,\Diamond}{k}{B}\in\PTIME$.
\end{theorem}

To prove this theorem, we present a polynomial-time algorithm deciding the problem. Because of Lemma~\ref{lemma:classes work multi-modal}, we can restrict ourselves to circuits from $\multcirc{\Box,\Diamond}{k}{\oplus,0,1}.$ First note that using $\Diamond_i,\oplus$, and the constant $1$, we can express $\Box_i$, and therefore it is sufficient to consider circuits in which only $\Diamond_i$-operators occur, i.e., we only need to deal with circuits from $\multcirc{\Diamond}{k}{\oplus,0,1}$.

The algorithm $\algname$ presented below decides this problem in polynomial time by converting circuits into a normal form. For a circuit $C$, let ${\algname}(C)$ denote the output of the algorithm $\algname$ when given $C$ as input. A decision algorithm derived from $\algname$ accepts a circuit $C$ if and only if ${\algname}(C)$ is not the constant $0$-circuit.

The intuitive approach of the algorithm is to delete redundant data, i.e., extra $0$s and sub-circuits corresponding to formulas of the form $\varphi\oplus\varphi$, which obviously are equivalent to $0,$ and to arrange the gates of the circuit in a standard order, to get a unique representation for the input circuit. In the propositional formula case, the approach is quite simple: For a formula in which only the operator $\oplus$, variables and constants appear, we repeatedly delete every variable or constant that appears twice, and remove $0$s. If this produces the empty formula or the formula containing only the constant $0$,  then the formula is unsatisfiable, otherwise it is satisfiable. Surprisingly, the generalization to modal logic and circuits instead of formulas performs only operations of a similarly simple type---however, proving the correctness requires more work than in the propositional case.

In the statement of the algorithm, the term $\Diamond$-gate refers to any $\Diamond_i$-gate for some $i\in\set{1,\dots,k}.$

\medskip

\begin{algorithmic}
  \STATE{\algname({\bf Input}: $C\in\multcirc{\Diamond}{k}{\set{\oplus,0,1}}$)}
  \WHILE{there are unmarked $\Diamond$-gates or the output gate is not marked}
    \STATE{Let $g$ be an unmarked $\Diamond$-gate such that all $\Diamond$-gates with a path to $g$ are marked if such a gate exists, let $g$ be the output gate otherwise.}
    \algspace
    \STATE{Let $G$ be the set of propositional gates before $g$ which are connected to $g$ with a path consisting only of propositional gates ($G$ includes $g$ if $g$ is propositional).}
    \algspace
    \STATE{Let $D_1,\dots,D_m$ be the subcircuits whose output gates are the $\Diamond$-gates directly before $G$.}
    \algspace
    \STATE{Consider $G$ as a propositional circuit with output gate $g$ and input gates $d_1,\dots,d_m$ replacing the subcircuits $D_1,\dots,D_m$.}
    \algspace
    \STATE{Rewrite $G$ as formula $\varphi:=d_{i_1}\oplus d_{i_2}\oplus\dots\oplus d_{i_j}\oplus\varphi'$, where each $d_i$ occurs at most once and where $\varphi'$ does not contain $d_1,\dots,d_m$} \\
    \algspace
    \WHILE{changes in $\varphi$ still occur}
      \algspace
      \STATE{Order $\varphi$ lexicographically.}
      \algspace
      \STATE{If $D_i$ and $D_j$ are identical, replace $d_i\oplus d_j$ with $0$.} \\
      \algspace
      \STATE{If $\mathcal{F}=\rm KD$, then replace $\Diamond_i 1$ with $1$ for any $i.$}
      \algspace
      \STATE{For any $i,$ replace $\Diamond_i 0$ with $0$.}\\
      \algspace
      \STATE{Remove $0$s unless the formula becomes empty.}
      \algspace
      \STATE{For propositional variable $p$, replace $p \oplus p$ with 0.}
      \algspace
      \STATE{Replace $1 \oplus 1$ with 0.}
      \algspace
    \ENDWHILE
    \algspace
    \STATE{Reintegrate $\varphi$ into the circuit, using connections from the $D_i$ subcircuits instead of the $d_i$ variables.}
    \algspace
    \STATE{mark $g$}
  \ENDWHILE
  \algspace
  \STATE{Delete gates not connected to the output gate.} \\
\end{algorithmic}

\medskip

We now show that the algorithm works correctly---note that the following theorem implies the correctness of the decision procedure outlined above, since $\algname$ returns $0$ when given a circuit consisting just of a $0$-gate as input.

\begin{theorem}\label{theorem:xor algorithm equivalence}
Let $C_1$ and $C_2$ be circuits from $\multcirc{\Diamond}{k}{\set{\oplus,0,1}}$, and let $\mathcal{F}\in\{\rm K,\rm KD\}.$ Then ${\algname}(C_1)={\algname}(C_2)$ if and only if $C_1\equiv_{\cal F}C_2.$
\end{theorem}

First we show that the algorithm can be implemented to work in polynomial time, and observe a useful property.

\begin{lemma}\label{appendix:xor algorithm is polynomial time}
The algorithm $\algname$ runs in polynomial time and satisfies ${\algname}({\algname}(C))={\algname}(C)$ for every circuit $C\in\multcirc {\Diamond}k{0,1,\oplus}$ for all $k\ge 1$.
\end{lemma}

\begin{proof}
We show that the algorithm works in polynomial time. The outer WHILE loop is run at most once for every gate in the circuit. The inner WHILE loop shortens the formula by at least one character in each iteration except one (where only sorting is performed). Each step in the algorithm can clearly be performed in polynomial time, the only non-obvious case is the ``Rewrite $G$ as formula'' step. This can be performed in polynomial time because propositional circuits representing linear functions can easily be converted into formulas:
Determine, by simulation, which of the variables is relevant for the function calculated by the circuit. The resulting formula consists of an XOR of all these variables and output value of the circuit when given zeros as input. Note that not all of the variables $d_1,\dots,d_m$ necessarily appear in the formula. The formula constructed in this way is at most as large as the original circuit.

Note that if $G$ already is a formula connected only to the output-gate $g$ and the $d_i$-gates, and $\varphi$ is lexicographically ordered, then the algorithm does not perform any changes at this step. This implies that ${\algname}({\algname}(C))={\algname}(C)$.
\end{proof}

We now prove a lemma needed in the correctness proof for the algorithm. The lemma states that for two XOR-formulas to be equivalent, two of the arguments to the XOR operators already have to be equivalent, and this enables us to give an inductive proof for Theorem~\ref{theorem:xor algorithm equivalence}.

\begin{lemma}\label{lemma:xor proof equivalent subcircuits}
Let $\mathcal{F}\in\{\K,\KD\}$, $k\ge 1$, $n\ge 2$,  $D_1,\dots,D_n\in\multcirc{\Diamond}{k}{\set{\oplus,0,1}}$, let $\varphi_1,\varphi_2$ be propositional XOR-formulas, and let $\Diamond_{i_1} D_1\oplus\dots\oplus\Diamond_{i_n} D_n\oplus\varphi_1\oplus\varphi_2$ be not $\mathcal{F}$-satisfiable, where $i_1,\dots,i_n\in\set{1,\dots,k}$. Then
\begin{enumerate}
\item If $\mathcal{F}=\K$ and all $D_i$ are $\mathcal{F}$-satisfiable, then there exist $1\leq i\neq j\leq n$ such that $D_i\equiv_\mathcal{F}D_j$.
\item If $\mathcal{F}=\KD$ and all $D_i$ are $\mathcal{F}$-satisfiable and not $\mathcal{F}$-tautologies, then there exist $1\leq i\neq j\leq n$ such that $D_i\equiv_\mathcal{F}D_j$.
\end{enumerate}
\end{lemma}

\begin{proof}
We first show that $\varphi_1$ is equivalent to $\varphi_2$ or to $\neg\varphi_2$. Consider an $\mathcal F$-model $M$ with a non-reflexive root world $w$. Changing truth assignments in $w$ only affects the propositional formulas $\varphi_1$ and $\varphi_2$. Since $\varphi_1\oplus\varphi_2$ is $\mathcal F$-equivalent to $\Diamond_{i_1} D_1\oplus\dots\oplus\Diamond_{i_n} D_n$, $\varphi_1\oplus\varphi_2$ must be constant. This only leaves these two choices for $\varphi_1,\varphi_2.$

Further, if $\mathcal{F}=\mathtext{K}$, then $\varphi_1\equiv\varphi_2$: Consider the frame $M$ with a world $w$ which does not have a successor. Since $\Diamond_{i_1} D_1\oplus\dots\oplus\Diamond_{i_n} D_n\oplus\varphi_1\oplus\varphi_2$ is not $\mathtext{K}$-satisfiable, this implies that $\varphi_1\oplus\varphi_2$ is not $\mathtext{K}$-satisfiable, thus $\varphi_1$ and $\varphi_2$ are $\mathtext{K}$-equivalent. Since these formulas are propositional, they are equivalent.

If $\varphi_1\equiv\varphi_2,$ then $\Diamond_{i_1} D_1\oplus\dots\oplus\Diamond_{i_n}D_n$ is not $\cal F$-satisfiable. If ${\cal F}=\mathtext{KD}$ and $\varphi_1\equiv\neg\varphi_2,$
then $\Diamond_{i_1} D_1\oplus\dots\oplus\Diamond_{i_n} D_n$ is an $\cal F$-tautology.
Assume that the $D_l$ are pairwise $\mathcal{F}$-inequivalent. Since all of the $D_l$ are $\cal F$-satisfiable, this implies that $n$ is even for $\varphi_1\equiv\varphi_2$, and odd for $\varphi_1\equiv\neg\varphi_2$: If this would not hold, we could construct a world which for every $D_i$ has an $i$-successor in which it holds, and this would satisfy $\Diamond_{i_1} D_1\oplus\dots\oplus\Diamond_{i_n} D_n\oplus\varphi_1\oplus\varphi_2.$

Let $D_m$ be a minimal element of $\{D_1,\dots,D_n\}$ with respect to $\cal F$-implication. This exists because $\cal F$-implication defines a partial order on the $D_i$ (the $\cal F$-inequivalence of the $D_i$ ensures the anti-symmetry). For each $l\neq k$, let $M_l$ be a model with a world $w_l$ such that $M_l,w_l\models D_l\wedge\neg D_m$. Let $M$ be a model containing a world $w$ which has all of the $w_l$ as successors. Then it holds that $\displaystyle M,w\models\neg \Diamond_{i_m} D_m\wedge\bigwedge_{l\neq k}\Diamond_{i_l} D_l,$ and thus $M,w$ satisfies an odd number of the $\Diamond_{i_m} D_m$ clauses if $n$ is even, and an even number if $n$ is odd. This model leads to a different truth value of the formula than the model where all of the $\Diamond_{i_m} D_m$'s are satisfied, which is a contradiction, because the formula is $\cal F$-constant.
\end{proof}

We now prove Theorem~\ref{theorem:xor algorithm equivalence}:

\begin{proof}
The \emph{propositional level} of a modal circuit $C$ with a propositional output gate is the set of propositional gates in $C$ that are connected to the output gate with a path having no gates representing modal operators.

Obviously, $C\equiv_{\cal F}{\algname}(C)$, and therefore ${\algname}(C_1)={\algname}(C_2)$ implies $C_1\equiv_{\cal F}C_2$. We now show the other direction.

Observe that the following holds when $\algname$ is given a circuit as input which on its propositional level is a formula (i.e., every gate in the propositional level has fan-out of at most $1$), which has circuits $D_i$ as inputs:

\begin{equation}\label{eq:l in p proof recursion} {\algname}(\Diamond_{i_1} D_1\oplus\dots\oplus\Diamond_{i_l} D_l)={\algname}(\Diamond_{i_1}{\algname}(D_1)\oplus\dots\oplus\Diamond_{i_l}{\algname}(D_l))
\end{equation}

Assume that the theorem does not hold, and let $C_1,C_2$ be $\cal F$-equivalent circuits, $l_i$ the number of diamonds in $C_i,$ such that ${\algname}(C_1)\neq{\algname}(C_2)$ and such that the pair $(C_1,C_2)$ is minimal with respect to $l_1+l_2$, and let $l_1\geq l_2.$ Because of Lemma~\ref{appendix:xor algorithm is polynomial time}, and since $\algname$ does not add diamonds, we can assume ${\algname}(C_1)=C_1$ and ${\algname}(C_2)=C_2.$

If the output gate of $C_1$ is propositional (without loss of generality, this is an $\oplus$-gate), then the algorithm converts the propositional level of the circuit to a formula over the variables corresponding
to the $\Diamond$-gates which are connected to the output gates with a non-modal path. Thus, since $\algname(C_1)=C_1$, we can consider the circuit as a formula $C_1=\Diamond_{i_1} D_1\oplus\dots\oplus\Diamond_{i_l} D_l\oplus\varphi_1,$ where the $D_l$ are the subcircuits starting before the highest diamonds. If the output gate of $C_1$ is modal, then $C_1$ is of the same form, with $k=l$ and $\varphi_1$ absent. In the same way, assume $C_2=\Diamond_{i_{l+1}} D_{l+1}\oplus\dots\oplus\Diamond_{i_n}
D_n\oplus\varphi_2$.

The circuits $D_1,\dots,D_l$ are pairwise $\cal F$-inequivalent: Assume $D_1\equiv_{\cal F} D_2$. Then, by minimality of $C_1,C_2,$ it holds that ${\algname}(D_1)={\algname}(D_2).$ Therefore, because the $D_j$ are lexicographically ordered (since ${\algname}(C_1)=C_1$), equation (\ref{eq:l in p proof recursion}) implies that $\Diamond_{i_1}{\algname}(D_1)\oplus\Diamond_{i_2}{\algname}(D_2)$ will be replaced with $0,$ which is a contradiction to ${\algname}(C_1)=C_1$. The same holds for $D_{l+1},\dots,D_n$. By an analogous argument, all of the $D_j$ are $\cal F$-satisfiable: $\algname$ converts unsatisfiable $D_j$ to $0$ and deletes them, since the $D_j$ have less diamonds than $l_1+l_2$. Additionally, if $\cal F=\rm KD$, we can assume that none of the $D_j$ is a KD-tautology, because $\Diamond_{i_j} 1$ is replaced by $1$.

Assume there exist $i,j$ such that $1\leq i\leq l<j\leq n,$ and $D_i\equiv_{\cal F} D_j$. By minimality of $l_1+l_2,$ it holds that ${\algname}(D_1)={\algname}(D_{l+1})$.  Define $E$ as $\algname(D_1)$.
Now, since we have

\medskip

\begin{tabular}{lcl}
$C_1={\algname}(C_1)$ & $=$ & $\phantom{\oplus}\Diamond_{i_1}{\algname}(D_1)\oplus\dots\oplus_{i_{i-1}}\Diamond{\algname}(D_{i_{i-1}})\oplus\Diamond_{i_1} E$ \\ && $\oplus\Diamond_{i_{{i+1}}} {\algname}(D_{i+1})\oplus\dots\oplus\Diamond_{i_l}{\algname}(D_{l})\oplus\left(\algname(\varphi_1)\right)$ \vspace*{2mm}\\ 
$C_2={\algname}(C_2)$ & $=$ & $\phantom{\oplus}\Diamond_{i_{l+1}}{\algname}(D_{l+1})\oplus\dots\oplus\Diamond_{i_{j-1}}{\algname}(D_{j-1})\oplus\Diamond E$ \\ && $\oplus\Diamond_{i_{j+1}}{\algname}(D_{j+1})\oplus\dots\oplus\Diamond_{i_n}{\algname}(D_n)\oplus\left(\algname(\varphi_2)\right)$,
\end{tabular}

\ \\
\noindent
we can replace $E$ with $0$ in $C_1$ and $C_2,$ and get a counter-example with less diamonds than $l_1+l_2,$ which is a contradiction. Therefore, all of the $D_j$ are pairwise $\cal F$-inequivalent and satisfiable. $C_1\oplus C_2=\Diamond_{i_1} D_1\oplus\dots\Diamond_{i_n} D_n\oplus\varphi_1\oplus\varphi_2$ is not $\cal
F$-satisfiable, since $C_1\equiv_{\cal F} C_2$. Thus, with Lemma~\ref{lemma:xor proof equivalent subcircuits} it follows that there exist $1\leq i\neq j\leq n$ such that $D_i\equiv_{\mathcal F} D_j$. This is a contradiction.

Thus, it follows that $n\leq 1$. First assume $n=1$. Then $C_1=\Diamond_{i_1} D_1\oplus\varphi_1\equiv_{\cal F}\varphi_2=C_2$. This is equivalent to $\varphi_1\equiv_{\cal F}\varphi_2$ ($\varphi_1\equiv_{\cal F}\neg\varphi_2$) and $D_1$ is not $\cal F$-satisfiable (an $\cal F$-tautology). Thus, $D_1$ is not $\cal F$-satisfiable (an $\cal F$-tautology), which is a contradiction to the above.

Therefore $n=0$, and both circuits are propositional (since any occurring $\Diamond$-gates that are not connected to the output-gate are removed by the algorithm), and $\algname(C_i)=C_i$. In this case, $\algname$ rewrites the input circuits as formulas, orders the appearing variables and constants, and deletes double occurrences. The result is a unique formula representation of the input circuit. Thus, the claim holds for $l=0$, and hence the theorem is proven.
\end{proof}

It is interesting to note that since the algorithm never adds a gate to a circuit, Theorem~\ref{theorem:xor algorithm equivalence} implies that for a given input circuit, the algorithm computes a smallest possible circuit representing the same function. This implies that minimization problems in this context can be decided in polynomial time as well. Note that the algorithm, when given a formula as input, also returns a formula. Therefore, it can be used to minimize both circuits and formulas.

The above proof does not generalize to other classes of frames. The main reason is that no analog of Lemma~\ref{lemma:xor proof equivalent subcircuits} seems to hold for classes of frames involving, for example, reflexivity or transitivity. While we conjecture that the corresponding problem for these classes of frames can still be solved in polynomial time, we mention that there are examples in the literature that behave differently---sometimes, restricting the class of frames increases the complexity of the modal satisfiability problem. For example, Halpern showed that when considering only formulas of bounded modal nesting degree, the complexity of the satisfiability problem for \K\ drops from \PSPACE-complete to \NP-complete. On the other hand, for the logic \Sfour, the problem remains \PSPACE-complete \cite{hal95}. Therefore syntactical restrictions that reduce the complexity of the general logic \K\ do not necessarily also reduce the complexity for logics defined over a restricted class of models.

Our results for linear propositional functions conclude our discussion about the modal satisfiability problem for the class of frames \K\ in the case that we allow both modal operators in our formulas and circuits: Figure~\ref{fig:lattice} shows that we have covered all clones, and hence reached a complete classification of this problem.

\subsection{Satisfiability With Only One Operator}

We now look at satisfiability problems with only one of type of operators $\Diamond$ or $\Box$ present. For sets $B$ such that $\clone B\supseteq {\rm S}_1$, we already established \PSPACE-completeness for the classes of frames we consider (Corollary~\ref{cor:s1 formulas pspace complete}). Since polynomial-time results for the case where we allow both $\Box$ and $\Diamond$ obviously carry over to the case where only one of them is allowed, the following theorem completes a full classification of the problem.

\begin{theorem}\label{theorem:monotone in p for one modal operator and k or k4}
  Let $B$ be a finite set of Boolean functions such that $B\subseteq\cM,$ let $k\ge 0$, and let $M=\set{\Box}$ or $M=\set{\Diamond}$, and let $\mathcal F\in\set{\K,\Kfour}.$ Then $\circsat{\K}{M}{k}{B}\in\PTIME$.
\end{theorem}

\begin{proof}
Due to Lemma~\ref{lemma:classes work multi-modal}, we can assume that $B=\set{\wedge,\vee,0,1}.$ We now show that in the case $M=\set{\Diamond},$ a circuit $C\in\multcirc{\Diamond}{k}{B}$ is $\mathcal F$-satisfiable if and only if it is satisfied in the reflexive singleton where each variable is set to true, and in the case $M=\set{\Box},$ every $C\in\multcirc{\Box}{k}{B}$ is $\mathcal F$-satisfiable if and only if it is satisfied in the irreflexive singleton with every variable set to true (since both the reflexive and the irreflexive singleton are $\mathcal F$-models, the ``if'' direction of this claim is trivial). These conditions obviously can be tested in polynomial time.

We show the claim by induction on the structure of the formula expansion of the circuit. If $C$ is a single variable or a constant, then the claim obviously holds. Now assume that $C=C_1\vee C_2.$ If $C$ is satisfiable, then at least one of $C_1,C_2$ is satisfiable, and due to induction, they are satisfied in the reflexive resp. irreflexive singleton with every variable set to true. If $C=C_1\wedge C_2,$ and $C$ is satisfiable then both $C_1$ and $C_2$ are satisfiable. By induction, both of them are satisfied in the singleton with every variable set to true. Hence, $C$ is satisfied in this singleton as well.

For the modal operators, assume that $C=\Diamond_i D$ for some $i\in\set{1,\dots,k}.$ If $C$ is satisfiable, then obviously $D$ is satisfiable as well, and by induction, $D$ is satisfiable in the reflexive singleton with every variable set to true. For this case, $C$ obviously is satisfied in the same model.

Finally, if $C=\Box_i D$ for some $i\in\set{1,\dots,k},$ then by definition $C$ is satisfied in the irreflexive singleton with every variable set to true.
\end{proof}

\section{The Validity Problem}\label{section:validity}

Besides the satisfiability problem, another problem which often is of interest is the validity problem, i.e., the problem to decide whether a given formula is valid, or is a tautology in a given logic. Recall that in our context, a formula $\varphi$ is an $\mathcal F$-tautology if and only if $\varphi$ is $\mathcal F$-equivalent to $1$ (this is the case if and only if $\varphi$ holds in every world of every $\mathcal F$-model).

It is obvious that a formula $\varphi$ is a tautology if and only if $\neg\varphi$ is not satisfiable. With this easy observation, the complexity of the satisfiability problem and that of the validity problem often can be related to each other---they are ``duals'' of each other. However, in the case of restricted propositional bases, we cannot always express negation, which is necessary in order to do the transformation mentioned above directly. Therefore, we consider a more general notion of duality, which is closely related to the self-dual property defined for functions earlier: A function $f$ is self-dual if and only if $\dual f=f.$ 

\begin{definition}
Let $f$ be an $n$-ary Boolean function. Then $\dual f$ is the $n$-ary Boolean function defined as $\dual f(x_1,\dots,x_n)=\neg f(\overline{x_1},\dots,\overline{x_n}).$
\end{definition}

For a set $B$ of Boolean functions, $\dual B$ is defined as the set $\set{\dual f\ \vert\ f\in B}.$ Obviously, a similar duality exists between the modal operators $\Diamond$ and $\Box$: For a set $M\subseteq\set{\Box,\Diamond}$, we define $\dual M$ to be the set such that $\Box\in\dual M$ if and only if $\Diamond\in M$, and $\Diamond\in\dual M$ if and only if $\Box\in M$. For a clone $B,$ the dual clone $\dual B$ can easily be identified in Post's Lattice (see Figure~\ref{fig:lattice}), as it is simply the ``mirror class'' with regard to the vertical symmetry axis in the lattice. The following theorem shows that complexity classifications for the satisfiability problem immediately give dual classifications for the validity problem.

\begin{theorem}
Let $B$ be a finite set of Boolean functions, let $k\ge 0$, and let $\mathcal F$ be a class of frames, and let $M\subseteq\set{\Box,\Diamond}.$ Then the following holds:
\begin{enumerate}
\item $\circtaut{\mathcal F}{M}{k}{B}\eqlogm\overline{\circsat{\mathcal F}{\dual M}{k}{\dual B}}.$
\item $\formtaut{\mathcal F}{M}{k}{B}\eqlogm\overline{\formsat{\mathcal F}{\dual M}{k}{\dual B}}.$
\end{enumerate}
\end{theorem}

\begin{proof}
Let $C$ be a circuit from $\multcirc{M}{k}{B}.$ We construct the circuit $\dual C$ by exchanging every $f$-gate for a function $f\in B$ with a $\dual f$-gate. Similarly, we replace every $\Box_i$-gate with a $\Diamond_i$-gate, and vice versa. It is obvious that this transformation can be performed in logarithmic space, and that the same transformation can be applied to formulas.

It remains to prove that $C$ is unsatisfiable if and only if $\dual C$ is a tautology. Since $\dual .$ is obviously injective, and $\dual{\dual C}=C,$ this also proves that $C$ is a tautology if and only if $\dual C$ is unsatisfiable, and hence proves the reduction.

Inductively, we show a more general statement: For any modal model $M,$ let $\neg M$ denote the model obtained from $M$ by reversing the propositional truth assignment, i.e., where a variable in a world is true if and only if the same variable is false in the same world in $M.$ We show that for any model $M$ and any world $w\in M,$ it holds that $M,w\models C$ if and only if $\neg M,w\nmodels\dual C.$ This obviously completes the proof, since $M$ is an $\mathcal F$-model if and only if $\neg M$ is.

We show the claim by induction on the structure of $C.$ First, assume that $C$ is equivalent to the variable $x_i.$ Then $M,w\models C$ if and only if $M,w\models x_i$ if and only if $\neg M,w\nmodels x_i.$ Since $\dual{x_i}=x_i,$ this proves the base step.

Now assume that the output gate $g$ of $C$ is an $f$-gate for an $n$-ary Boolean function $f\in B,$ and let $g_1,\dots,g_n$ be the predecessor gates of $g$ in $C.$ By induction, we know that for each $j\in\set{1,\dots,n},$ it holds that $M,w\models C_{g_j}$ if and only if $\neg M,w\nmodels\dual {C_{g_j}}$ (where $C_{g_j}$ is the subcircuit of $C$ with output gate $g_j$). For $j\in\set{1,\dots,n},$ let $\alpha_j$ be defined as $1$ if $M,w\models C_{g_j},$ and $0$ otherwise. By induction, we know that $\alpha_j$ is $1$ if and only if $\neg M,w\nmodels\dual C.$ Now observe that $M,w\models C$ if and only if $f(\alpha_1,\dots,\alpha_n)=1,$ if and only if $\dual{f}(\overline{\alpha_1},\dots,\overline{\alpha_n})=0,$ and this is the case if and only if $\neg M,w\nmodels\dual C.$

Now assume that the output gate $g$ of $C$ is a $\Diamond_i$-gate for some $i\in\set{1,\dots,k},$ and let $h$ be the predecessor gate of $g$ in $C.$ Then the following holds:

\smallskip

\begin{tabular}{rcl}
$M,w\models C$ & iff & there is a world $w'$ such that $(w,w')\in R_i$ and $M,w'\models C_h$ \\
               & iff & there is a world $w'$ such that $(w,w')\in R_i$ and $\neg M,w'\nmodels \dual{C_h}$\\
               & iff & $\neg M,w\nmodels\Box_i\dual{C_h}$\\
               & iff & $\neg M,w\nmodels\dual C.$
\end{tabular}

\smallskip

Finally, assume that the output gate $g$ of $C$ is a $\Box_i$-gate for some $i\in\set{1,\dots,k},$ and let $h$ be the predecessor gate of $g$ in $C.$ Then the following holds:

\smallskip

\begin{tabular}{rcl}
$M,w\models C$ & iff & for each world $w'$ such that $(w,w')\in R_i,$ $M,w'\models C_h$ \\
               & iff & for each world $w'$ such that $(w,w')\in R_i,$ $\neg M,w'\nmodels \dual{C_h}$\\
               & iff & $\neg M,w\nmodels\Diamond_i\dual{C_h}$\\
               & iff & $\neg M,w\nmodels\dual C.$
\end{tabular}

This concludes the induction, and therefore the proof.
\end{proof}

\section{Conclusion and Further Research}\label{section:conclusion}

We completely classified the complexity of the modal satisfiability and validity problems arising when restricting the allowed propositional operators in the formula for the logics \K\ and \KD. We showed that the more succinct representation of modal formulas as circuits does not have an effect on the complexity of these problems up to $\redpm$-degree. We also showed that for multi-modal logics, the results only depend on whether we have $0$, $1$, or $2$ modalities, adding more modal operators does not increase the complexity of the problems we studied. Note that in many cases, our results hold for more general classes of frames, as often, they are stated for any class containing the reflexive singleton, or similar conditions. This does not only apply to most of our polynomial-time results, but also for our circuit-to-formula construction in Corollary~\ref{corollary:circuit-to-formula}, and our implementation results in Theorem~\ref{theorem:s11 implementation} and the uni-modal version of Theorem~\ref{theorem:s1 implementation}.

The most obvious next question to look at is to complete our complexity classification for other classes of
frames. For $\mathcal F\in\{\mathtext{T},\mathtext{S}4,\mathtext{S}5\}$,
our proofs already give a complete classification with the exception of the
complexity of the problems $\formsat{\mathcal F}{M}{k}{B}$ and $\circsat{\mathcal F}{M}{k}{B}$ where $\clone B\in\set{\cL_0,\cL_1}$. We conjecture that these cases are solvable in polynomial time as well, however, to solve these cases different ideas from the ones used in the proof for \K\ and \KD\ are required. Another interesting question is the exact complexity of our polynomial cases, most notably the case where the propositional operators represent linear functions. 

There are many other interesting directions for future research. For example, one can look at other decision problems (e.g., global satisfiability and formula minimization), and one can try to generalize modal logic modally as well as propositionally.

\section{Acknowledgments}

We thank Michael Bauland for his work on the work presented in~\cite{bhss05b}, and Thomas Schneider and Heribert Vollmer for helpful discussions. We also thank the anonymous STACS referees for their helpful comments and suggestions, and Steffen Reith for providing the figure of Post's Lattice.

\newcommand{\etalchar}[1]{$^{#1}$}

\end{document}